\documentclass[12pt,english]{article}
\usepackage[T1]{fontenc}
\usepackage[latin9]{inputenc}
\usepackage{color}
\usepackage{babel}
\usepackage{refstyle}
\usepackage{mathtools}
\usepackage{url}
\usepackage{dsfont}
\usepackage{amsmath}
\usepackage{amsthm}
\usepackage{amssymb}
\usepackage{geometry}
\usepackage{setspace}
\usepackage[authoryear]{natbib}
\usepackage[bookmarks=true,bookmarksnumbered=true,bookmarksopen=true,bookmarksopenlevel=1,
 breaklinks=false,pdfborder={0 0 0},pdfborderstyle={},backref=false,colorlinks=true]
 {hyperref}
\hypersetup{
 citecolor=blue, linkcolor=blue}

\makeatletter

\AtBeginDocument{\providecommand\secref[1]{\ref{sec:#1}}}
\AtBeginDocument{\providecommand\eqref[1]{\ref{eq:#1}}}
\AtBeginDocument{\providecommand\propref[1]{\ref{prop:#1}}}
\AtBeginDocument{\providecommand\appxref[1]{\ref{appx:#1}}}
\AtBeginDocument{\providecommand\defref[1]{\ref{def:#1}}}
\AtBeginDocument{\providecommand\axmref[1]{\ref{axm:#1}}}
\AtBeginDocument{\providecommand\lemref[1]{\ref{lem:#1}}}
\AtBeginDocument{\providecommand\claimref[1]{\ref{claim:#1}}}
\AtBeginDocument{\providecommand\subsecref[1]{\ref{subsec:#1}}}
\AtBeginDocument{\providecommand\thmref[1]{\ref{thm:#1}}}
\providecommand{\tabularnewline}{\\}
\RS@ifundefined{subsecref}
  {\newref{subsec}{name = \RSsectxt}}
  {}
\RS@ifundefined{thmref}
  {\def\RSthmtxt{theorem~}\newref{thm}{name = \RSthmtxt}}
  {}
\RS@ifundefined{lemref}
  {\def\RSlemtxt{lemma~}\newref{lem}{name = \RSlemtxt}}
  {}

\numberwithin{equation}{section}
\theoremstyle{definition}
\newtheorem{defn}{\protect\definitionname}
\theoremstyle{plain}
\newtheorem{ax}{\protect\axiomname}
\theoremstyle{plain}
\newtheorem{thm}{\protect\theoremname}
\theoremstyle{plain}
\newtheorem{prop}{\protect\propositionname}
\theoremstyle{plain}
\newtheorem{assumption}{\protect\assumptionname}
\theoremstyle{plain}
\newtheorem{lem}{\protect\lemmaname}
\theoremstyle{remark}
\newtheorem{claim}{\protect\claimname}

\usepackage{amssymb}
\usepackage{refstyle}
\usepackage{accents}
\usepackage{charter}%{mathdesign} %requires ly1 if fails
\usepackage{paralist}
\usepackage{dsfont}

\def\sloppy{%
  \tolerance 1000%
  \emergencystretch 1em%
  \hfuzz .5\p@
  \vfuzz\hfuzz}
\newref{sec}{name=Section~} \newref{subsec}{name=Subsection~} \newref{axm}{name=Axiom~} \newref{lem}{name=Lemma~} \newref{def}{name=Definition~} \newref{prop}{name=Proposition~} \newref{thm}{name=Theorem~} \newref{remark}{name=Remark~} \newref{cor}{name=Corollary~} \newref{fig}{name=Figure~} \newref{eq}{name=Equation~} \newref{appx}{name=Appendix~} \newref{oappx}{name=Online Appendix~} \newref{fn}{name=Footnote~} \newref{fact}{name=Fact~} \newref{exa}{name=Example~} \newref{cond}{name=Condition~} \newref{claim}{name=Claim~}

\makeatother

\providecommand{\assumptionname}{Assumption}
\providecommand{\axiomname}{Axiom}
\providecommand{\claimname}{Claim}
\providecommand{\definitionname}{Definition}
\providecommand{\lemmaname}{Lemma}
\providecommand{\propositionname}{Proposition}
\providecommand{\theoremname}{Theorem}

\begin{document}
\title{Informative Consumption}
\author{Xuehan Jiang\thanks{Department of Economics, University of Toronto (Email: xuehan.jiang@mail.utoronto.ca).}~~~~~~~~~Xi
Zhi ``RC'' Lim\thanks{Antai College of Economics and Management, Shanghai Jiao Tong University
(Email: xzlim@sjtu.edu.cn).}}
\date{\vskip -.8emThis version: 2026/06/15\\
Most recent public version: \url{http://s.xzlim.com/ic}{\small}}
\maketitle
\begin{abstract}
Risky consumption generates information when uncertainty is resolved.
This paper axiomatically characterizes the consumption\textendash information
trade-off even when the analyst does not observe an agent\textquoteright s
future problems. A subjective future menu underpins the agent\textquoteright s
willingness to sacrifice current consumption for future information.
By carefully separating objective risk from subjective risk, we decompose
the certainty equivalent of an act into a standard risk premium and
a novel information premium. To facilitate applications, we introduce
an Arrow\textendash Debreu\textendash Pratt parameterization that
yields a tractable model, capturing risk aversion and information
incentives with a single coefficient for each. Finally, we show that
heterogeneity in risk-taking may arise from differing opportunities
to capitalize on information, rather than being solely attributable
to differences in risk aversion.
\end{abstract}
\newpage{}
\begin{center}
\textbf{Table of contents for working paper.}
\par\end{center}

\tableofcontents{}
\begin{center}
\textbf{Table of contents for working paper.}
\par\end{center}

\newpage{}

\section{Introduction}

Economic agents often have to choose an action not knowing for certain
what the outcome will be. These risky consumption decisions inherently
contain information when uncertainty is resolved, rewarding the agent
not only in terms of today's utility (consumption incentive) but also
through better future decisions (information incentive). Disentangling
the two incentives helps address why population heterogeneity in risk-taking
behavior is not solely the result of heterogeneous risk attitudes,
but also the result of heterogeneous future opportunities, where some
individuals can benefit more from (the same) information than other
individuals do.

Information incentives are ubiquitous, even if the analyst does not
observe it directly. Consider the well-known \emph{Anscombe\textendash Aumann
act}, the predominant tool to study choice under (subjective) risk.
It prescribes the probability of different outcomes $x,y,z,...$ under
each state of the world $\omega_{1},\omega_{2},...$, general enough
to find applications across all of economics. For example, the act
$f$ is risky but $g$ always delivers the same outcome $z$.
\begin{center}
\begin{tabular}{c|cc|}
\multicolumn{1}{c}{$f$} & $x$ & \multicolumn{1}{c}{$y$}\tabularnewline
\cline{2-3}
$\omega_{1}$ & $0.9$ & $0.1$\tabularnewline
$\omega_{2}$ & $0.3$ & $0.7$\tabularnewline
\cline{2-3}
\end{tabular}\hskip 2em%
\begin{tabular}{c|c|}
\multicolumn{1}{c}{$g$} & \multicolumn{1}{c}{$z$}\tabularnewline
\cline{2-2}
$\omega_{1}$ & $1$\tabularnewline
$\omega_{2}$ & $1$\tabularnewline
\cline{2-2}
\end{tabular}
\par\end{center}

\noindent When the DM does not know the true state $\omega$, she
forms a belief about the likelihood of different states and outcomes,
and uses that belief to make consumption decisions. However, upon
choosing $f$, outcome $x$ is substantially more likely to occur
under state $\omega_{1}$ than under state $\omega_{2}$, so when
uncertainty is resolved and $x$ occurs, a Bayesian DM receives new
information suggesting that the state is quite likely $\omega_{1}$.
This information may have value depending on what the DM's future
problems entail, even if they fall outside the scope of the analyst\textquoteright s
observations.

In fact, an Anscombe\textendash Aumann act is structurally identical
to a \emph{Blackwell experiment}, and $f$ Blackwell-dominates the
uninformative $g$. The DM may hence have an incentive to choose $f$
even when its consumption utility is lower than that of $g$, and
analogously an incentive to choose $g$ even though it is Blackwell-dominated
by $f$. This simple observation is consistent with a wide range of
real-life problems. A diner trying a new restaurant learns about its
quality while enjoying the meal. A doctor administering a treatment
observes the patient's clinical response and learns about the treatment's
effectiveness. A firm deploying a marketing strategy will learn about
features of consumer demand. The extent to which (the same) information
has value depends on unobserved individual characteristics, whether
the diner will soon move out of town, whether the doctor is aware
of new future treatments, whether the firm has a certain confidential
expansion plan.

To disentangle the unobserved consumption\textendash information trade-off,
we axiomatically characterize a model called \textbf{Informative Consumption
(IC)}. The primitive is simply preference over Anscombe\textendash Aumann
acts. In the model, the DM has a \emph{subjective future menu}, $A$,
which is ultimately inferred from choice behavior. Along with other
usual preference parameters, the DM chooses act $f$ to maximize
\begin{equation}
\alpha U_{C}\left(f\right)+\left(1-\alpha\right)U_{I}\left(f\right),\label{eq:intro}
\end{equation}
where $U_{C}\left(f\right)$ is the standard, first-period expected
consumption utility, $U_{I}\left(f\right)$ is a continuation value
from information acquisition, and $\alpha$ is a consumption-information
weight. The continuation value $U_{I}\left(f\right)$ depends on both
the information content in $f$ and the options in $A$. The model
is formally given in \secref{setup}.

One interpretation when $\alpha\leq0.5$ is discounting, where the
DM enjoys an immediate consumption utility and a discounted future
utility. The maximization of \eqref{intro} is equivalent to the maximization
of
\[
U_{C}\left(f\right)+\beta U_{I}\left(f\right)
\]
where $\beta=\left(1-\alpha\right)/\alpha\in\left[0,1\right]$.

IC admits the standard (one-period) Subjective Expected Utility (SEU)
as a special case when $\alpha=1$ (or $\beta=0$). But unlike essentially
all generalizations of SEU (where the focus has predominantly been
on the violations of expected utility forms or the presence of multiple
priors), IC departs from SEU not due to behavioral biases. Quite the
contrary, the DM is sophisticated, forward-looking, and takes information
incentives into account.

In fact, the main axiom \emph{Mixture Aversion} is the exact opposite
of the well-known maxmin expected utility's axiom Uncertainty Aversion
\citep{gilboa1989maxmin}. When the DM is indifferent $f\sim g$,
mixing the acts weakens the information content of signal realizations
because they are now more equally distributed across states. Hence
$f\succ af+\left(1-a\right)g$ when the DM is not indifferent to the
mixture.

\subsubsection*{Risk\textendash Information Decomposition}

IC enables the study of a novel notion called \emph{information premium}
that is in many ways analogous to the well-known \emph{risk premium}.
In this framework, both premiums play a role in determining a DM's
certainty equivalent of an act. Because a certainty equivalent is
certain, the (degenerate act) that offers the certainty equivalent
is uninformative, so the DM must be compensated for the loss of information
value as she gives up an act. We show that information premium has
intuitive properties. The information premium of an uninformative
act $f$ is zero. Under an intuitive parameterization, \propref{riskpremiuminformationpremium}
further shows that Blackwell-dominating acts have greater information
premiums. It also shows that information premiums weakly increase
for all acts when information incentive is greater, which is intuitive
because there is more valuable use of the same information. The difference
between an act's expected payoff and its certainty equivalent is exactly
the sum of risk premium and information premium, both of which can
be identified by carefully separating objective risk (risky but no
information) from subjective risk (contains information).

\subsubsection*{Behavioral Characterization}

The characterization of IC involves four key axioms (and a technical
condition). \nameref{axm:informationseeking} is already mentioned.
\nameref{axm:non-overlapping} relaxes the standard independence condition
by applying it only to mixtures that involve proportional likelihoods
in generating the same prize across different states. Intuitively,
the new axiom allows for a DM to care about whether a prize $z$ occurs
through state $\omega_{1}$ or state $\omega_{2}$, for instance,
which is not the case in standard SEU where the DM simply calculates
the aggregate probability of $z$. Of course, mixing proportional
acts does not alter information content because posteriors are determined
by likelihood ratios. The last two significant axioms are\emph{ }\nameref{axm:prize_independence}
and \nameref{axm:state_prize_independence}. The former weakens monotonicity
and says that an ``upgraded'' act should be preferred if the upgrade
does not involve overlapping prizes; again, the intuition is about
not distorting the likelihood of a prize across different states.
The latter requires that equal and state-dependent distortions do
not change preferences.

\subsubsection*{Arrow\textendash Debreu\textendash Pratt Parameterization}

For applications, \secref{info_coef} introduces a 2-parameter parameterization
of IC where $\rho$ captures risk aversion and $\iota$ captures information
incentive. It is common to work with CARA functional forms where the
Arrow\textendash Pratt coefficient $\rho$ captures the degree of
risk aversion. It turns out that the future menu $A$ can also have
an intuitive parameterization using a menu of Arrow\textendash Debreu
assets, each paying the prize $\iota$ in a particular state. An individual
with a greater $\iota$ has stronger information incentive and always
enjoys more information premiums. We show that this parameterization
greatly simplifies IC. In particular, continuation value is now separable
between preference parameters $\left(\rho,\iota\right)$ and a novel
information index $\mathbb{I}\left(f,\mu\right)$ that captures the
aggregate information value of $f$ under prior $\mu$ (\propref{icad}).
Results show that differences in risk-taking behavior can be reconciled
using the same degree of risk aversion but with the more risk-taking
DM having a greater information coefficient $\iota$ if her chosen
act has a greater information index (\propref{same_rho}).

~

\secref{setup} introduces the primitive, the model, and the risk\textendash information
decomposition. \secref{Axioms} proposes axioms and state the representation
theorem. \secref{info_coef} proposes the Arrow\textendash Debreu\textendash Pratt
representation. \secref{Conclusion} concludes. \appxref{Appendix:-Proofs}
contains proofs.

\subsection{Related Literature}

In the subjective learning literature that can be traced back to \citet{kreps1979representation},
theories use observed preferences over menus to infer the possibility
that a DM expects to receive information tomorrow, even if that information
is not observed by the analyst. Most work focuses on \textbf{exogenous
subjective learning}, where the DM has no control over what she might
learn \citep{dekel2001representing,dillenberger2014theory}. Then
comes \textbf{endogenous subjective learning}, in which the DM has
some control \citep{hyogo2007subjective,piermont2016learning,cooke2017preference}.
In this literature, learning is subjective because the analyst does
not observe the information structures of the DMs.

Our paper studies a different problem where information structures
are objectively known because they are Anscombe\textendash Aumann
acts. Instead, our main unobservable is the (future) menu that the
DM expects to face, which is observed in the subjective learning literature.

There is also rising interest in the \textbf{intrinsic preference
for information}, where a DM has preference over the resolution of
uncertainty even if that information cannot help future decisions
\citep{masatlioglu2023intrinsic,gul2025thrill,ke2025preference}.
In contrast, we study extrinsic preference for information because
they can improve future decisions. In theory, decision making can
be affected by both forms of preference for information. Our work
shows that distinguishing objective risk from subjective risk will
separate the two even if extrinsic information incentives are not
directly observed.

There is also a literature on \textbf{pure information acquisition}
without consumption, in which we are most related to two sub-strands.
One sub-strand characterizes \textbf{costless information acquisition}
\citep{gilboa1991value,azrieli2008value,frankel2019quantifying}.
We share technical methods because the need to infer a DM's unobserved
problem is in common. Our study disentangles information and consumption
incentives in the presence of unobserved information\textendash consumption
trade-off. The second sub-strand takes \textbf{information cost} into
the picture. \citet{caplin2015revealed} characterize optimal information
acquisition, \citet{de2017rationally} identify information acquisition
cost, and \citet{caplin2022rationally} characterize Shannon cost.
Unlike our study, decision problems are observed. We note that our
approach may contribute to microfounding information cost using consumption
loss.

\section{\label{sec:setup}Model}

Let the set of prizes be $X=\mathbb{R}$, and let $\Omega=\left\{ \omega_{1},...,\omega_{N}\right\} $
be a finite state space with $\left|\Omega\right|=N$. Let $\Delta\left(X\right)$
be the space of simple lotteries, i.e., lotteries with finite support.
An (Anscombe\textendash Aumann) act is a function $f:\Omega\rightarrow\Delta\left(X\right)$
and $\mathcal{F}$ is the set of all acts. The primitive is a binary
relation $\succsim$ over $\mathcal{F}$. The symmetric and asymmetric
components of $\succsim$ are denoted by $\sim$ and $\succ$, respectively.

\textbf{Notation}: Per convention, $p\left(x\right)$ denotes the
probability that lottery $p$ gives prize $x$ and $f\left(x\mid\omega\right)$
denotes the probability that act $f$ gives prize $x$ under state
$\omega$. A \emph{degenerate lottery} $\delta_{x}$ gives prize $x$
with probability 1. A\emph{ constant act} $p\in\mathcal{F}$ gives
lottery $p\in\Delta\left(X\right)$ in every state. A \emph{degenerate
act} $\delta_{x}\in\mathcal{F}$ is a constant act that gives the
degenerate lottery $\delta_{x}$. For $f,g\in\mathcal{F}$ and $a\in\left[0,1\right]$,
the \emph{mixture act} $af+\left(1-a\right)g$ denotes the act that
gives prize $x$ under state $\omega$ with probability $af\left(x\mid\omega\right)+\left(1-a\right)g\left(x\mid\omega\right)$.
As is standard, $\text{supp}\left(f\right)\coloneqq\left\{ x\in X:\sum_{\omega}f\left(x\mid\omega\right)>0\right\} $
for all $f\in\mathcal{F}$ and $\text{supp}\left(p\right)\coloneqq\left\{ x\in X:p\left(x\right)>0\right\} $
for all $p\in\Delta\left(X\right)$. For notational brevity, we use
$f\cap g$ to refer to $\text{supp}\left(f\right)\cap\text{supp}\left(g\right)$.
For later, given $A\subseteq\mathcal{F}$, let $\text{supp}\left(A\right):=\cup_{f\in A}\text{supp}\left(f\right)$.
We use the convention that $0\times-\infty$ is $0$.

\subsection{\label{subsec:Model}Informative Consumption Representation}

Even though our observation is preference over (one period's) acts,
the DM in our model is, as if, facing a two-period problem. Given
a prior over states $\mu\in\Delta\left(\Omega\right)$ and a Bernoulli
utility function $u:X\rightarrow\mathbb{R}$, her choice of act $f$
gives her the standard (subjective) expected utility in the \textbf{first
period},
\begin{equation}
U_{C}\left(f\right)=\sum_{\omega\in\Omega}\mu\left(\omega\right)\sum_{x\in\text{supp}\left(f\right)}f\left(x\mid\omega\right)u\left(x\right).\label{eq:Uc}
\end{equation}
In the \textbf{second period}, the DM solves an (unobserved) problem
utilizing the information obtained from the resolution of $f$, which
gives her an (ex-ante) continuation value 
\begin{equation}
U_{I}\left(f\right)\coloneqq\sum_{x\in\text{supp}\left(f\right)}\underbrace{\left[\sum_{\omega\in\Omega}\mu\left(\omega\right)f\left(x\mid\omega\right)\right]}_{\text{probability of signal/prize }x}\underbrace{\left[\max_{g\in A}\sum_{\omega\in\Omega}\mu_{x,f}\left(\omega\right)\sum_{y\in\text{supp}\left(g\right)}g\left(y\mid\omega\right)u\left(y\right)\right]}_{\text{value of posterior \ensuremath{\mu_{x,f}}}}.\label{eq:UI}
\end{equation}
The formulation of continuation value is standard: Upon receiving
a prize $x$, the DM forms a Bayesian posterior $\mu_{x,f}\in\Delta\left(\Omega\right)$,
i.e.,
\[
\mu_{x,f}\left(\omega\right)=\frac{\mu\left(\omega\right)f\left(x\mid\omega\right)}{\sum_{\omega'\in\Omega}\mu\left(\omega'\right)f\left(x\mid\omega'\right)}\,\forall\omega.
\]
Then, she chooses from her perceived future menu $A$ to maximize
expected utility, which results in the value of having posterior $\mu_{x,f}$.
Finally, by aggregating over possible prize realizations $x\in\text{supp}\left(f\right)$,
she obtains her continuation value.
\begin{defn}
\label{def:IC_statement}$\succsim$ admits an \textbf{Informative
Consumption (IC) }representation if there exist a full support prior
$\mu\in\Delta\left(\Omega\right)$, a strictly increasing utility
function $u:X\rightarrow\mathbb{R}$, a consumption-information weight
$\alpha>0$, and a future menu $A\subseteq\mathcal{F}$ such that
$\succsim$ is represented by the function $U:\mathcal{F}\rightarrow\mathbb{R}$
where
\[
U\left(f\right)=\alpha U_{C}\left(f\right)+\left(1-\alpha\right)U_{I}\left(f\right).
\]
\end{defn}
The model departs from standard SEU, but it is not due to behavioral
biases. Instead, the DM is more sophisticated in the sense that she
considers not only today's consumption but the information benefits
on future decisions. Note that $A\subset\mathcal{F}$ is a \emph{subjective
future menu}. It is neither decided nor known by the analyst, as it
is not part of the primitive. Instead, it captures the DM's perception
of what the future entails. It could be correct but falls outside
of the analyst's observation, or it could simply be the DM's misconception.
In any case, this perceived future menu captures an extrinsic preference
for information that is hidden to the analyst.

The special case where $\alpha=1$ corresponds to a DM who only cares
about consumption value, and we obtain standard SEU. But even if $\alpha<1$,
information incentive can still be absent when $A$ is a singleton,
that is, the DM does not have a choice to make in the future. In this
case, information may be present but it has no value, and in particular
$U_{I}\left(f\right)=\text{constant}$ for all $f$. Clearly, maximizing
$U\left(f\right)=U_{C}\left(f\right)+\text{constant}$ is standard
SEU. \secref{Axioms} characterizes IC and related representations.
\secref{info_coef} parameterizes $A$ with a (single) coefficient
of information incentive.

\subsection{\label{subsec:information_premium}Risk\textendash Information Decomposition}

Risk premium has served as an incredibly useful measure to summarily
capture an economic agent's risk aversion. We propose a novel notion
of \emph{information premium} because risk-taking behavior can also
be affected by information incentives.

The \emph{certainty equivalent} of $f$, denoted by $CE\left(f\right)$,
is a prize in $X$ such that receiving that sure prize (in every state)
is as good as receiving $f$, i.e., $\delta_{CE\left(f\right)}\sim f$.
For a given distribution of prizes, standard methods infer that a
lower certainty equivalent is associated with greater risk aversion.

This is no longer true in our model. The certainty equivalent can
still be well-defined, but it now accounts for two distinct forces.
The first is standard. Because certainty is preferred by a risk averse
DM, a certainty equivalent tends to be lower than the mean of an act
where the difference is usually called the \emph{risk premium}. This
is the only force for a constant act $p\in\mathcal{F}$ (it gives
the lottery $p\in\Delta\left(X\right)$ in every state of the world),
so we obtain the standard
\[
\mathbb{E}_{p}\left(x\right)-CE\left(p\right)=\text{risk premium}.
\]

The second force is novel. Because the degenerate act $\delta_{CE\left(f\right)}$
is uninformative, the certainty equivalent of $f$ must use pure consumption
to compensate for forgone information utility from $f$, so it leads
to a greater certainty equivalent due to what we call \emph{information
premium}.

A typical act $f\in\mathcal{F}$ involves both premiums and we propose
a decomposition. Let $p_{f}$ be the \emph{lottery reduction} of $f$,
i.e., $p_{f}\left(x\right)=\sum_{x}\mu\left(\omega\right)f\left(x\mid\omega\right)$.
Consider
\begin{align}
\mathbb{E}_{f}\left(x\right)-CE\left(f\right) & =\underbrace{\left[\mathbb{E}_{f}\left(x\right)-CE\left(p_{f}\right)\right]}_{\text{RP}\left(f\right)\text{ (risk premium)}}+\underbrace{\left[CE\left(p_{f}\right)-CE\left(f\right)\right]}_{-\text{IP}\left(f\right)\text{ (information premium)}}\label{eq:decompose}
\end{align}
where $\mathbb{E}_{f}\left(x\right)=\sum_{\omega}\mu\left(\omega\right)f\left(x\mid\omega\right)x$.

Rearranging \eqref{decompose} gives
\begin{equation}
CE\left(f\right)=\mathbb{E}_{f}\left(x\right)-\text{RP}\left(f\right)+\text{IP}\left(f\right).\label{eq:cef}
\end{equation}

Information premium $\text{IP}$ compensates the DM for giving up
information. Therefore, it summarily captures the influence of a DM's
information incentive on her evaluation of risky prospects. It is
non-negative. If $f$ is a constant act, then $f$ is identical to
its reduction $p_{f}$ and $\text{IP}\left(f\right)=0$. This is intuitive
since a constant act does not provide information. We show in \propref{riskpremiuminformationpremium},
under a parameterization, that information premium is analogous to
the familiar risk premium in the way it reacts to preference parameters
and the nature of different acts.

\subsection{Information and Subjective Risk}

The decomposition into different premiums captures how \emph{objective
risk} and \emph{subjective risk} differently affect the evaluation
of an act.

Objective risk refers to the uncertainty from an objective distribution
of prizes, such as the lottery given in a particular state $f\left(\omega\right)$\emph{.
}Differently, subjective risk refers to the uncertainty due to the
DM's subjective prior about states interacted with state-dependent
outcomes. For instance, act $f$ has subjective risk but no objective
risk, and act $g$ has objective risk but no subjective risk.
\begin{center}
\begin{tabular}{c|cc|}
\multicolumn{1}{c}{$f$} & $x_{1}$ & \multicolumn{1}{c}{$x_{2}$}\tabularnewline
\cline{2-3}
$\omega_{1}$ & $1$ & $0$\tabularnewline
$\omega_{2}$ & $0$ & $1$\tabularnewline
\cline{2-3}
\end{tabular}\hskip 2em%
\begin{tabular}{c|cc|}
\multicolumn{1}{c}{$g$} & $x_{1}$ & \multicolumn{1}{c}{$x_{2}$}\tabularnewline
\cline{2-3}
$\omega_{1}$ & $1/4$ & $3/4$\tabularnewline
$\omega_{2}$ & $1/4$ & $3/4$\tabularnewline
\cline{2-3}
\end{tabular}
\par\end{center}

Standard models in economics, such as the SEU, evaluate both forms
of risks using expected utility. In the example above, if $\mu\left(\omega_{1}\right)=1/4$,
then the lottery reductions $p_{f},p_{g}$ are identical and SEU deems
$f,g$ indifferent. The presence of information incentive causes a
departure from SEU, but it only does so through subjective risk.

In the decomposition (\eqref{decompose}), because $p_{f}$ is a constant
act containing only objective risk, which means it is structurally
identical to an uninformative Blackwell experiment, risk premium is
never affected by subjective risk.

On the other hand, $p_{f}$ and $f$ have the same aggregate probabilities
of prizes, so they should lead to the same consumption utilities.
Their differences come from the subjective risk contained in $f$
(interacting with $\mu$), which provides information when resolved.
Therefore, information premium is the result of the interaction between
information incentive and subjective risk.

If $A$ is a singleton (the DM thinks that she has no future choice
to make), we return to standard SEU. But even if $A$ is meaningfully
non-trivial, constant acts can never contain information because they
do not contain subjective risk, their information premiums are still
0.

\secref{info_coef} continues this discussion under an intuitive parameterization.

\subsubsection*{Information-induced risk}

We remark that there is a necessary trade-off between information
and risk. The reduction $p_{f}$ of $f$ is degenerate if and only
if there exists $z\in X$ such that $f\left(z\mid\omega\right)=1$
for all $\omega\in\text{supp}\mu$. In that case, $\text{\text{IP}}\left(f\right)=0$
and $\text{RP}\left(f\right)=0$.

If $\text{\text{IP}}\left(f\right)>0$, then $p_{f}$ is non-degenerate.
If $u$ is strictly risk averse (i.e., strictly concave), then $\text{\text{IP}}\left(f\right)>0$
implies $\text{RP}\left(f\right)>0$. This is intuitive because information
can only be transmitted if an act can generate different prizes, or
signal realizations, under a given prior. In turn, this departure
from a degenerate act induces risk premium for risk averse DMs. The
converse is not true since $\text{RP}\left(f\right)>0$ can be the
result of a constant but non-degenerate act, which is risky but uninformative,
and hence $\text{\text{IP}}\left(f\right)=0$.

\section{\label{sec:Axioms}Axioms and Representation Theorem}

\begin{ax}[Basic Axioms]
\label{axm:standard}~
\begin{enumerate}
\item (Basic Rationality) The binary relation $\succsim$ is complete and
transitive.
\item (Continuity) If $f\succ g\succ h$, then $\left\{ a:af+\left(1-a\right)h\succsim g\right\} $
and $\left\{ a:g\succsim af+\left(1-a\right)h\right\} $ are closed.
\item (Monotonicity) If $x>y$, then $\delta_{x}\succ\delta_{y}$.\footnote{This can be replaced by a weaker, non-extremality condition (i.e.,
for all $x\in X$, neither $\delta_{x}\succsim\delta_{y}$ for all
$y\in X$ nor $\delta_{y}\succsim\delta_{x}$ for all $y\in X$) along
with the corresponding weakening of $u$ in \defref{IC_statement}
and \defref{IC_V}.}
\end{enumerate}
\end{ax}

\subsection{Skewness Independence}

The benchmark model we depart from is Subjective Expected Utility
(SEU), and its key axiom is the independence condition. The full scale
of independence requires, indiscriminately, $f\succsim g\Leftrightarrow af+\left(1-a\right)h\succsim ag+\left(1-a\right)h$.
The next axiom weakens this requirement by adding a prerequisite related
to overlapping prizes. Let $f_{x}\coloneqq\left(f\left(x\mid\omega_{1}\right),...,f\left(x\mid\omega_{N}\right)\right)$
capture the likelihood of prize $x$ in different states. Given acts
$h,f\in\mathcal{F}$, write $h_{x}\propto f_{x}$ if $kh_{x}=f_{x}$
for some $k\geq0$. That is, the likelihood of prize $x$ occurring
through different states is proportional across the two acts (examples
later).
\begin{ax}[Skewness Independence]
\label{axm:non-overlapping}If $h_{x}\propto f_{x}$ and $h_{x}\propto g_{x}$
for all $x\in\text{\text{supp}}\left(h\right)$, then $f\succsim g\Leftrightarrow af+\left(1-a\right)h\succsim ag+\left(1-a\right)h$.
\end{ax}
A standard SEU maximizer considers only the aggregate probability
of each prize occurring regardless of its potential skewness across
states. \axmref{non-overlapping} relaxes this condition by allowing
for a DM who cares about whether a prize is more likely to occur through
state $\omega_{1}$ or through state $\omega_{2}$. The mixture of
overlapping acts can alter this likelihood ratio. In the following
examples, the prizes of $h$ do not overlap with the prizes of $f$
and $g$, so independence holds between $\left(f,g,h\right)$. For
$g'$ and $g''$, $x_{1}$ is now an overlapping prize with $h$.
But $h_{x}\propto g_{x}'$, so independence should still hold between
$\left(f,g',h\right)$. Yet the mixture $\frac{1}{2}g''+\frac{1}{2}h$
distorts the relative likelihood of $x$ by a shift from $\omega_{1}$
to $\omega_{2}$, so the axiom allows independence to fail (i.e.,
$f\succ g''$ but $\frac{1}{2}g''+\frac{1}{2}h\succ\frac{1}{2}f+\frac{1}{2}h$).
\begin{center}
\begin{tabular}{c|cc|}
\multicolumn{1}{c}{$h$} & $x$ & \multicolumn{1}{c}{$y$}\tabularnewline
\cline{2-3}
$\omega_{1}$ & $1/2$ & $1/2$\tabularnewline
$\omega_{2}$ & $1/2$ & $1/2$\tabularnewline
\cline{2-3}
\end{tabular}\hskip 2em%
\begin{tabular}{c|cc|}
\multicolumn{1}{c}{$f$} & $w$ & \multicolumn{1}{c}{$z$}\tabularnewline
\cline{2-3}
$\omega_{1}$ & $1$ & $0$\tabularnewline
$\omega_{2}$ & $0$ & $1$\tabularnewline
\cline{2-3}
\end{tabular}\hskip 2em%
\begin{tabular}{c|cc|}
\multicolumn{1}{c}{$g$} & $w$ & \multicolumn{1}{c}{$z$}\tabularnewline
\cline{2-3}
$\omega_{1}$ & $0$ & $1$\tabularnewline
$\omega_{2}$ & $1$ & $0$\tabularnewline
\cline{2-3}
\end{tabular}
\par\end{center}

\begin{center}
\begin{tabular}{c|cc|}
\multicolumn{1}{c}{$g'$} & $x$ & \multicolumn{1}{c}{$z$}\tabularnewline
\cline{2-3}
$\omega_{1}$ & $1/4$ & $3/4$\tabularnewline
$\omega_{2}$ & $1/4$ & $3/4$\tabularnewline
\cline{2-3}
\end{tabular}\hskip 2em%
\begin{tabular}{c|cc|}
\multicolumn{1}{c}{$g''$} & $x$ & \multicolumn{1}{c}{$z$}\tabularnewline
\cline{2-3}
$\omega_{1}$ & $5/6$ & $1/6$\tabularnewline
$\omega_{2}$ & $1/6$ & $5/6$\tabularnewline
\cline{2-3}
\end{tabular}
\par\end{center}

\subsection{Prize Independence}

Another standard axiom that is weakened is monotonicity. In most models,
``upgrading'' an act by replacing an inferior prize with a superior
prize results in the new act deemed better. This assumption is sometimes
called first-order stochastic dominance. Given act $f$, let $fxp$
be the act that replaces prize $x$ with lottery $p$. Formally, $fxp$
yields the prize $y\ne x$ with the original probability of $y$ plus
the added probability of $y$ due to the replacement of $x$, that
is, $f\left(y\mid\omega\right)+f\left(x\mid\omega\right)p\left(y\right)$.
It yields prize $x$ with probability $f\left(x\mid\omega\right)p\left(x\right)$.\footnote{$fxp\left(y\mid\omega\right)=\begin{cases}
f\left(y\mid\omega\right)+f\left(x\mid\omega\right)p\left(y\right) & y\neq x\\
f\left(x\mid\omega\right)p\left(x\right) & y=x
\end{cases}$.}

If $p\succsim x$, then $fxp$ is an ``upgrade'' from $f$, and
standard models predict $fxp\succsim f$. This is no longer obvious
if the DM cares not only about the aggregate probability of a prize
but also about its probability in each state. In particular, the replacement
of $x$ can distort the relative probability of $y$ between states.
To allow for monotonicity violations only in the face of these distortions,
observe that distortions will not occur if $\text{supp}\left(p\right)\cap\text{supp}\left(f\right)=\emptyset$,
because $f\left(x\mid\omega\right)p\left(y\right)>0$ only if $f\left(y\mid\omega\right)=0$.
Hence we impose monotonicity that pertains to an upgrade with this
precondition. For monotonicity that pertains to a downgrade, i.e.,
$x\succsim p\Rightarrow f\succsim fxp$, we impose without additional
preconditions.
\begin{ax}[Prize Independence]
\label{axm:prize_independence}~
\begin{itemize}
\item If $x\succsim p$ (resp. $\succ$), then $f\succsim fxp$ (resp. $\succ$).
\item If $p\cap f=\emptyset$ and $p\succsim x$ (resp. $\succ$), then
$fxp\succsim f$ (resp. $\succ$).
\end{itemize}
\end{ax}
Furthermore, if two prizes $x,y$ are equally likely under every
state, i.e., $f_{x}=f_{y}$, then equal distortions to both prizes
should preserve indifferences. Formally, consider the weaker notion
of replacing a prize where it is replaced only in a particular state.
Given $f$, let $f_{\omega}xp$ be the result of replacing prize $x$
with lottery $p$ (only) under state $\omega$.\footnote{$f_{\omega}xp\left(y\mid\omega\right)=\begin{cases}
fxp\left(y\mid\omega'\right) & \omega'=\omega\\
f\left(y\mid\omega'\right) & \omega'\ne\omega
\end{cases}$.} We demand that two different ways of modifying $f$ by replacing
equally likely prizes with non-overlapping, indifferent lotteries
do not make one modification better than another.
\begin{ax}[State-Prize Independence]
\label{axm:state_prize_independence}If $f_{x}=f_{y}$, $\delta_{x}\sim p$,
$\delta_{y}\sim q$, $p\cap f=\ensuremath{\emptyset}$, and $q\cap f=\ensuremath{\emptyset}$,
then $f_{\omega}xp\sim f_{\omega}yq.$
\end{ax}

\subsection{Mixture Aversion}

The final axiom concerns the mixture of two indifferent acts. Under
standard independence, the mixture of two acts is indifferent to the
two acts themselves. The DMs we consider are mixture averse. Intuitively,
the mixture operation tends to balance the relative probabilities
of a prize under different states. Recall that \axmref{non-overlapping}
and \axmref{prize_independence} avoid imposing restrictions when
these relative probabilities are potentially distorted. Mixture aversion
focuses on a particular kind of reaction to these distortions, where
making the same prize more-or-less likely to come from any state is
disliked.
\begin{ax}[Mixture Aversion]
\label{axm:informationseeking}If $f\sim g$, then $f\succsim af+\left(1-a\right)g$,
for any $a\in\left[0,1\right]$.
\end{ax}
To interpret, the DM prefers a strong association of each prize to
a particular state, which may be viewed as a particular form of state-dependent
preference for prizes. This is intuitive in the context of information
acquisition, where it is well known that mixing two Blackwell experiments
yields \textquotedblleft balanced\textquotedblright{} likelihood ratios,
thereby producing a posterior distribution that is closer to the prior,
which is less valuable for the subsequent decision problem.

Many models of subjective risk focus on the exact opposite axiom called
uncertainty aversion, where mixture is weakly preferred (i.e., $af+\left(1-a\right)g\succsim f$).\footnote{\citet{sarver2018dynamic} studies a different version of mixture
aversion where the primitive is preference over \citet{epstein1989substitution}
lotteries, and EU-violating behavior is understood as inconsistent
risk attitudes.} \citet{gilboa1989maxmin} is most prominent, where mixtures are preferred
and they underpin ambiguity averse preferences. In our setting, a
certain form of uncertainty may be preferred because the resolution
of these uncertainty contains useful information.

\subsection{\label{subsec:Representation}Representation Theorem}

An intermediate representation bypasses the future menu $A$ by directly
specifying a value function $V$ over posteriors:
\begin{defn}
\label{def:IC_V}$\succsim$ admits an \textbf{Informative Consumption
- V (IC-V)} representation if there exist a full support prior $\mu\in\Delta\left(\Omega\right)$,
a strictly increasing utility function $u:X\rightarrow\mathbb{R}$,
a convex (posterior) value function $V:\Delta\left(\Omega\right)\rightarrow\mathbb{R}$
such that $\succsim$ is represented by 
\[
U\left(f\right)=U_{C}\left(f\right)+\sum_{x\in\text{supp}\left(f\right)}\underbrace{\sum_{\omega\in\Omega}\mu\left(\omega\right)f\left(x\mid\omega\right)}_{\text{probability of posterior \ensuremath{\mu_{x,f}}}}V\left(\mu_{x,f}\right).
\]
\end{defn}
\begin{thm}
\label{thm:IC-1}$\succsim$ satisfies \nameref{axm:standard}, \nameref{axm:non-overlapping},
\nameref{axm:prize_independence}, and \nameref{axm:informationseeking}
if and only if it admits an Informative Consumption (IC) - V representation.
\end{thm}
\begin{thm}
\label{thm:IC-2}$\succsim$ admits an Informative Consumption (IC)
- V representation $\left(\mu,u,V\right)$ where $V$ is Lipschitz
continuous if and only if it admits an Informative Consumption (IC)
representation $\left(\mu,u,\alpha,A\right)$.\footnote{Lipschitz continuous: there exists $L\geq0$ such that for all $\pi,\pi'\in\Delta\left(\Omega\right)$,
$\left|V\left(\pi\right)-V\left(\pi'\right)\right|\leq L\left|\pi-\pi'\right|$.}
\end{thm}
We remark that the Lipschitz continuity of $V$ can be obtained by
adding a continuity axiom à la \citet{dekel2007representing}, where
we essentially require that the equivalents of any two acts that are
sufficiently close (in conditional probabilities) have bounded differences.
We leave this attempt for future versions of this paper.

\subsubsection*{Identification}

In both IC and IC-V, the prior $\mu$ is unique and the Bernoulli
utility function $u$ is unique up to positive affine transformation,
as in standard SEU. For IC-V, the value function $V$ can scale with
$u$, which is standard. Moreover, because $V$ can be ``tilted''
without changing the ex-ante value of any distribution of posteriors
(i.e., the continuation value), we express this flexibility using
the dot product $\lambda\cdot\pi$ in \axmref{prize_independence}.
\begin{prop}
\label{prop:identificationV}Suppose $\succsim$ admits IC-V representation
$\left(\mu_{1},u_{1},V_{1}\right)$, then it also admits an IC-V representation
$\left(\mu_{2},u_{2},V_{2}\right)$ if and only if $\mu_{1}=\mu_{2}$
and there exist $\kappa\in\mathbb{R}_{++}$, $\lambda\in\mathbb{R}^{\left|\Omega\right|}$,
and $a,b\in\mathbb{R}$ such that $u_{2}=\kappa u_{1}+a$ and $V_{2}\left(\pi\right)=\kappa V_{1}\left(\pi\right)+\lambda\cdot\pi+b$
for all $\pi\in\Delta\left(\Omega\right)$.
\end{prop}
The identification of IC is affected by the inevitable non-uniqueness
of the future menu $A$. For instance, the inclusion of dominated
acts in $A$ will not affect decisions. Moreover, multiplicity of
optimal acts in $A$ will induce redundancy. What can be identified,
instead, is that admissible $A$'s must induce the same (posterior)
value function:
\begin{defn}
Given $\left(\mu,u,\alpha,A\right)$, let the induced (posterior)
value function $V^{*}:\Delta\left(\Omega\right)\rightarrow\mathbb{R}$
be
\begin{equation}
V^{*}\left(\mu_{x,f}\right)=\left(\frac{1-\alpha}{\alpha}\right)\sup_{g\in A}\sum_{\omega\in\Omega}\mu_{x,f}\left(\omega\right)\sum_{y\in\text{supp}\left(g\right)}g\left(y\mid\omega\right)u\left(y\right).\label{eq:V*}
\end{equation}
\end{defn}
\begin{prop}
\label{prop:identificationA}Suppose $\succsim$ admits IC representation
$\left(\mu_{1},u_{1},\alpha_{1},A_{1}\right)$, then it also admits
an IC representation $\left(\mu_{2},u_{2},\alpha_{2},A_{2}\right)$
if and only if $\mu_{1}=\mu_{2}$ and there exist $\kappa\in\mathbb{R}_{++}$,
$\lambda\in\mathbb{R}^{\left|\Omega\right|}$, and $a,b\in\mathbb{R}$
such that $u_{2}=\kappa u_{1}+a$ and $V_{2}^{*}\left(\pi\right)=\kappa V_{1}^{*}\left(\pi\right)+\lambda\cdot\pi+b$
for all $\pi\in\Delta\left(\Omega\right)$.
\end{prop}
Both identification results utilize the fact that every posterior
belief is possible can be induced because $\mu$ has full support
and $\mathcal{F}$ is the set of all acts/experiments.

\section{\label{sec:info_coef}Arrow\textendash Debreu Information Coefficient}

In general, the future menu $A$ can involve arbitrarily many parameters.
For applications and comparative statics, the most important feature
of $A$ may be a DM's degree of information incentive. We propose
a one-dimensional measure that captures this incentive and its heterogeneity
across DMs.
\begin{assumption}
Throughout \secref{info_coef}, $u$ is continuous .
\end{assumption}

\subsection{Arrow\textendash Debreu Menu}

Drawing inspiration from the Arrow\textendash Debreu asset, let $a_{\iota_{i},\omega_{i}}\in\mathcal{F}$
be an act that pays $\iota_{i}\geq0$ units of consumption in state
$\omega_{i}$ and $0$ in other states. We also normalize for $u\left(0\right)=0$.
We call $A$ an Arrow\textendash Debreu menu if $A=\left\{ a_{\iota_{1},\omega_{1}},a_{\iota_{2},\omega_{2}},...,a_{\iota_{N},\omega_{N}}\right\} $.
Given an Arrow\textendash Debreu menu $A$, we call $\iota_{i}$ the
coefficient of information incentive for state $\omega_{i}$. A DM
stands to gain more by choosing $a_{\iota_{i},\omega_{i}}$ if $\iota_{i}$
is greater, and $a_{\iota_{i},\omega_{i}}$ is the optimal act if
the DM is sufficiently certain that the state is $\omega_{i}$, hence
$\iota_{i}$ indirectly captures the DM's incentive to distinguish
state $\omega_{i}$ from other states. If the coefficients are the
same across states, i.e., $\iota_{i}=\iota$ for all $i$, we say
that $A$ has a constant Arrow\textendash Debreu coefficient of information
incentive $\iota$. All else the same, a DM with a greater $\iota$
benefits more from information. Consider the following special case
of IC.
\begin{defn}
\label{def:mu_u_i}We say $\succsim$ has an \textbf{Arrow\textendash Debreu
(AD)} representation $\left(\mu,u,\iota\right)$ if it has an IC representation
$\left(\mu,u,\alpha,A\right)$ where $A$ has a constant Arrow\textendash Debreu
coefficient of information incentive $u^{-1}\left(\frac{\alpha}{1-\alpha}u\left(\iota\right)\right)$.

When $\alpha=0.5$, the coefficient of information incentive is simply
$\iota$. In general, the consumption weight $\alpha$ also affects
a DM's information incentive, where a greater $\alpha$ (more weight
on consumption) counterbalances a greater $\iota$ (more benefits
from information). Therefore, and without loss of generality, the
AD representation absorbs the effect of $\alpha$ into $\iota$ so
that the single parameter $\iota$ captures total information incentive.
\end{defn}
\begin{prop}
\label{prop:icad}If $\succsim$ has an AD representation $\left(\mu,u,\iota\right)$,
then $\succsim$ is represented by $U\left(f\right)=U_{C}\left(f\right)+U_{I}\left(f\right)$
where
\[
U_{I}\left(f\right)=u\left(\iota\right)\mathbb{I}\left(f,\mu\right)
\]
and
\begin{equation}
\mathbb{I}\left(f,\mu\right)\coloneqq\sum_{x\in\text{supp}\left(f\right)}\max_{\omega\in\Omega}\mu\left(\omega\right)f\left(x\mid\omega\right).\label{eq:index}
\end{equation}
\end{prop}
In particular, $U_{I}\left(f\right)$ (previously \eqref{UI}) is
now separable in $\left(u,\iota\right)$ and $\left(f,\mu\right)$
through a novel notion of \emph{information index} $\mathbb{I}$,
which has intuitive properties:\footnote{(1) $\mathbb{I}\left(f,\mu\right)=\sum_{x}\max_{\omega}\mu\left(\omega\right)\mathds{1}_{x\sim\omega}=\sum_{\omega}\mu\left(\omega\right)=1$.
(2) $\mu\left(\omega\right)=\frac{1}{\left|\Omega\right|}\rightarrow0$
so $\mathbb{I}\left(f,\mu\right)\rightarrow0$. (3) and (4) are due
to $g$ being the product of a Markov kernel applied to $f$'s signal
realizations, which weakens the values attained by max.}
\begin{enumerate}
\item $\mathbb{I}\left(f,\mu\right)=1$ if $f$ is fully informative.
\item $\mathbb{I}\left(f,\mu\right)\rightarrow0$ if $f$ is uninformative,
$\mu$ is uniform, and $\left|\Omega\right|\rightarrow\infty$. 
\item $\mathbb{I}\left(f,\mu\right)\geq\mathbb{I}\left(g,\mu\right)$ if
$f$ is Blackwell more informative than $g$.
\item $\mathbb{I}\left(f,\mu\right)>\mathbb{I}\left(g,\mu\right)$ if $f$
is strictly Blackwell more informative than $g$.
\end{enumerate}
If $f$ is a constant act, then $\mathbb{I}\left(f,\mu\right)$ reduces
to
\begin{equation}
\mathbb{N}\left(\mu\right)\coloneqq\max_{\omega\in\Omega}\mu\left(\omega\right)\text{.}\label{eq:N}
\end{equation}
A high $\mathbb{N}\left(\mu\right)$ corresponds to the situation
where the DM is a priori confident about the true state.

\subsubsection*{Information incentive, subjective risk, information premium}

The constant acts $p_{f}$ and $\delta_{z}$ have the same information
index $\mathbb{I}\left(f,\mu\right)=\mathbb{N}\left(\mu\right)$,
so $U_{I}\left(p_{f}\right)=U_{I}\left(\delta_{z}\right)=u\left(\iota\right)\mathbb{N}\left(\mu\right)$.
This is despite the fact that $p_{f}$ may still contain objective
risk. The result is $CE\left(p_{f}\right)$ simply solves $u\left(CE\left(p_{f}\right)\right)=U_{C}\left(p_{f}\right)$,
so risk premium $\text{RP}\left(f\right)$ is never affected by information
incentive $\iota$.

The informative incentive $\iota$ influences $CE\left(f\right)$
through information premium, which in the AD representation simplifies
to\footnote{By solving $U_{I}\left(\delta_{CE\left(f\right)}\right)=u\left(z\right)+u\left(\iota\right)\left(\max_{\omega}\mu\left(\omega\right)\right)\equiv\sum_{x}p_{f}\left(x\right)u\left(x\right)+u\left(\iota\right)\mathbb{I}\left(f,\mu\right)=U_{I}\left(f\right)$
for $CE\left(f\right)=u^{-1}\left(U_{C}\left(f\right)+u\left(\iota\right)\left[\mathbb{I}\left(f,\mu\right)-\mathbb{N}\left(\mu\right)\right]\right)$.} 
\begin{equation}
\text{\text{IP}}\left(f\right)=u^{-1}\left(U_{C}\left(f\right)+u\left(\iota\right)\left[\mathbb{I}\left(f,\mu\right)-\mathbb{N}\left(\mu\right)\right]\right)-u^{-1}\left(U_{C}\left(f\right)\right).\label{eq:ip-AD}
\end{equation}
Note that information premium is increasing in the degree of information
incentive $\iota$, increasing in the information index $\mathbb{I}\left(f,\mu\right)$
of $f$, and decreasing in the DM's ex-ante confidence about the true
state as measured by $\mathbb{N}\left(\mu\right)$, all of which are
intuitive.

In the AD representation, if $f$ is itself a constant act or the
prior $\mu$ is degenerate, then $\mathbb{I}\left(f,\mu\right)=\mathbb{N}\left(\mu\right)$.
\eqref{ip-AD} then gives$\text{\text{IP}}\left(f\right)=0$. Or if
$\iota=0$, then \eqref{ip-AD} gives $\text{IP}\left(\cdot\right)=0$.
In other words, pure objective risk and the lack of information incentive
each lead to standard SEU behavior, illustrating the discussion in
\secref{setup}.
\begin{prop}
\label{prop:riskpremiuminformationpremium}For any given $f\in\mathcal{F}$,
all else the same,
\begin{enumerate}
\item $\text{\emph{RP}}\left(f\right)$ weakly increases when $u$ is more
risk-averse,\footnote{As in standard, we say $u$ is more risk averse than $v$ if there
exists a concave $h:\mathbb{R}\rightarrow\mathbb{R}$ such that $u=h\circ v$.}
\item $\emph{\text{IP}}\left(f\right)$ weakly increases in $\iota$,
\item $\mathcal{\text{\emph{RP}}}\left(f\right)\geq\text{\emph{RP}}\left(g\right)$
if $p_{f}$ is a mean-preserving spread of $p_{g}$,
\item $\text{\text{\emph{IP}}}\left(f\right)\geq\text{\text{\emph{IP}}}\left(g\right)$
if $f$ Blackwell-dominates $g$ and $CE\left(p_{f}\right)=CE\left(p_{g}\right)$.
\end{enumerate}
All weak notions become strict if the sufficient conditions become
strict.

\end{prop}
\begin{proof}
Items 1 and 3 are well-known. Items 2 and 4 are straightforward from
\eqref{ip-AD} and the properties of $\mathbb{I}$.
\end{proof}
The special case of risk neutrality (i.e., $u\left(x\right)=\lambda x$,
$\lambda>0$) simplifies \eqref{ip-AD} to
\[
\text{\text{IP}}\left(f\right)=\iota\left[\mathbb{I}\left(f,\mu\right)-\mathbb{N}\left(\mu\right)\right],
\]
which is independent of the exact prizes $f$ sends.

\subsection{\label{subsec:Arrow-Debreu-Pratt}Arrow\textendash Debreu\textendash Pratt
Representation}

It is much more common in applications to work with a parameterized
Bernoulli utility function like the one associated with the constant
absolute risk aversion, 
\[
u_{\rho}\left(x\right)=1-e^{-\rho x},
\]
where $\rho$ is the (constant) Arrow\textendash Pratt coefficient
of absolute risk aversion.

Paired with our constant Arrow-Debreu coefficient of information incentive
$\iota$, we obtain a 2-parameter model given in \defref{adp}. Clearly,
the results and observations about AD representations apply.
\begin{defn}
\label{def:adp}We say $\succsim$ has an \textbf{Arrow\textendash Debreu\textendash Pratt
(ADP)} representation $\left(\mu,\rho,\iota\right)$ if it has an
AD representation $\left(\mu,u_{\rho},\iota\right)$.
\end{defn}
In standard expected utility, the coefficient $\rho$ provides a simple
estimate of an individual's risk aversion. A higher $\rho$ is always
associated with a higher risk premium; that is, the individual is
willing to accept a lower sure prize (certainty equivalent) to give
up the same act, holding prior fixed. Therefore, if two individuals
with the same prior have different certainty equivalents for the same
act, it must be that the individual with a lower certainty equivalent
has a greater $\rho$, i.e., is more risk averse.

This is no longer true. Rewriting \eqref{cef} with preference parameters
included (holding $\mu$ fixed),
\begin{align*}
CE_{\rho_{1},\iota_{1}}\left(f\right) & =\mathbb{E}_{f}\left(x\right)-\text{RP}_{\rho_{1}}\left(f\right)+\text{\text{IP}}_{\rho_{1},\iota_{1}}\left(f\right)\\
CE_{\rho_{2},\iota_{2}}\left(f\right) & =\mathbb{E}_{f}\left(x\right)-\text{RP}_{\rho_{2}}\left(f\right)+\text{\text{IP}}_{\rho_{2},\iota_{2}}\left(f\right).
\end{align*}
Note that $CE_{\rho_{1},\iota_{1}}\left(f\right)=CE_{\rho_{2},\iota_{2}}\left(f\right)$
can be consistent with $\rho_{1}>\rho_{2}$. In that case, $\text{RP}_{\rho_{1}}\left(f\right)>\text{RP}_{\iota_{2}}\left(f\right)$
but $\text{\text{IP}}_{\rho_{1},\iota_{1}}\left(f\right)>\text{\text{IP}}_{\rho_{2},\iota_{2}}\left(f\right)$.

That is, even though individual 1 is more risk averse than individual
2, which should lead to the willingness to accept a lower sure prize,
she also benefits more from information relative to individual 2 and
therefore sacrifices more from accepting a sure prize. When the two
opposite forces exactly cancel each other out, her certainty equivalent
is the same as individual 2's.

More generally, if an individual takes less risk by choosing $g$
over $f$ and another individual takes more risk by choosing $f$
over $g$, standard methods infer that the former is more risk averse.
What we emphasize is that the two individuals could well have the
same risk aversion $\rho_{1}=\rho_{2}$ but differ in their information
incentive $\iota_{1}\ne\iota_{2}$. It turns out that the information
index $\mathbb{I}$ (\eqref{index}) underpins this indeterminacy.

Let $\succsim_{\left(\mu,\rho,\iota\right)}$ denote the preference
under an ADP representation $\left(\mu,\rho,\iota\right)$ and $\succ_{\left(\mu,\rho,\iota\right)}$
its asymmetric component (strict preference). Recall that $\succsim_{\left(\mu,\rho,0\right)}$
reduces to standard SEU.
\begin{prop}
\label{prop:same_rho}Suppose
\[
\begin{array}{ccc}
g\succ_{\left(\mu,\rho_{1},0\right)}f &  & f\succ_{\left(\mu,\rho_{2},0\right)}g\end{array}
\]
for some $\rho_{1}>\rho_{2}$. There exist $\bar{\rho}$ and $\iota_{1}<\iota_{2}$
such that 
\[
\begin{array}{ccc}
g\succ_{\left(\mu,\bar{\rho},\iota_{1}\right)}f &  & f\succ_{\left(\mu,\bar{\rho},\iota_{2}\right)}g\end{array}
\]
if and only if $\mathbb{I}\left(f,\mu\right)>\mathbb{I}\left(g,\mu\right)$.
\end{prop}

\subsubsection*{Joint identification of parameters}

\propref{same_rho} highlights potential indeterminacy in preference
parameters when an act has both greater aggregate risk and a greater
information index. Joint identification is still possible; just not
with these acts. Counter-intuitively at first, using acts that do
not ``dominate'' one another does not necessarily provide identification,
because different combinations of risk aversion and information incentive
can cancel each other out.

Joint identification of $\rho$ and $\iota$ can be achieved by carefully
separating objective risk and subjective risk. In what follows, we
compare two constant acts to identify $\rho$, compare two acts with
the same degree of information to identify $\mu$, and compare a fully
informative act with a constant act to identify $\iota$.

Recall that constant acts have an information premium of $0$. Therefore,
if $g$ is a non-degenerate constant act and the elicited certainty
equivalent is $z$, then $\rho^{*}$ simply solves $\mathbb{E}_{p_{g}}\left(u_{\rho}\left(x\right)\right)=u_{\rho}\left(g\right)$,
where $p_{g}=g\left(\omega_{1}\right)$ is independent of the prior.

Next, suppose $x_{1},x_{2}>0$ and the DM is found indifferent between
$h$ and $h_{z}$.
\begin{center}
\begin{tabular}{c|cc|}
\multicolumn{1}{c}{$h$} & $x_{1}$ & \multicolumn{1}{c}{$x_{2}$}\tabularnewline
\cline{2-3}
$\omega_{1}$ & $1$ & $0$\tabularnewline
$\omega_{2}$ & $0$ & $1$\tabularnewline
\cline{2-3}
\end{tabular}\hskip 2em%
\begin{tabular}{c|cc|}
\multicolumn{1}{c}{$h_{z'}$} & $z'$ & \multicolumn{1}{c}{$0$}\tabularnewline
\cline{2-3}
$\omega_{1}$ & $1$ & $0$\tabularnewline
$\omega_{2}$ & $0$ & $1$\tabularnewline
\cline{2-3}
\end{tabular}
\par\end{center}

\noindent\begin{flushleft}
Since both are fully informative ($\mathbb{I}\left(h,\mu\right)=\mathbb{I}\left(h_{z'},\mu\right)=1$),
$\mu^{*}$ uniquely solves
\[
\mu^{*}\left(\omega_{1}\right)u_{\rho^{*}}\left(x_{1}\right)+\left(1-\mu^{*}\left(\omega_{1}\right)\right)u_{\rho^{*}}\left(x_{2}\right)\equiv\mu^{*}\left(\omega_{1}\right)u_{\rho^{*}}\left(z'\right).
\]
If there are more than two states, a workaround that maintains $\mathbb{I}\left(h,\mu\right)=\mathbb{I}\left(h_{z'},\mu\right)$
is to assign a third distinct prize to both $h$ and $h_{z'}$ for
every other state. This identifies $\mu\left(\omega_{2}\right)/\mu\left(\omega_{1}\right)$.
The procedure can be repeated for every $\mu\left(\omega_{i}\right)/\mu\left(\omega_{1}\right)$
to eventually identify $\mu$.\footnote{For instance if $\Omega=\left\{ \omega_{1},\omega_{2},\omega_{3},\omega_{4}\right\} $,
then $h\left(\omega_{3}\right)=h_{z'}\left(\omega_{3}\right)=h\left(\omega_{4}\right)=h_{z'}\left(\omega_{4}\right)=\delta_{x_{3}}$,
so $\mathbb{I}\left(h,\mu\right)=\mu\left(\omega_{1}\right)+\mu\left(\omega_{2}\right)+\max_{\tilde{\omega}\in\left\{ \omega_{3},\omega_{4}\right\} }\mu\left(\tilde{\omega}\right)$
and $\mathbb{I}\left(h_{z'},\mu\right)=\mu\left(\omega_{1}\right)+\mu\left(\omega_{2}\right)+\max_{\tilde{\omega}\in\left\{ \omega_{3},\omega_{4}\right\} }\mu\left(\tilde{\omega}\right)$,
so we are back to $U_{I}\left(h\right)=U_{I}\left(h_{z'}\right)$.}
\par\end{flushleft}

Finally, let $f$ be a fully informative act and suppose its elicited
certainty equivalent is $z''$, then $\iota^{*}$ uniquely solves
\[
u_{\rho^{*}}\left(z''\right)+u_{\rho^{*}}\left(\iota\right)\mathbb{N}\left(\mu^{*}\right)\equiv\mathbb{E}_{p_{f}}\left(u_{\rho^{*}}\left(x\right)\right)+u_{\rho^{*}}\left(\iota\right)\mathbb{I}\left(f,\mu^{*}\right)
\]
 (\propref{icad}) where $\mathbb{I}\left(f,\mu\right)=1$ because
$f$ is fully informative and $p_{f}$ is the lottery reduction of
$f$ under $\mu^{*}$.

\section{\label{sec:Conclusion}Conclusion}

The focus of this paper is the inference of a DM's \emph{subjective
future menu}, which entails the DM's perception of future problems
and the need to sacrifice consumption today for information tomorrow.

Our axiomatic characterization exploits the behavioral signatures
of a DM\textquoteright s preferences over Anscombe\textendash Aumann
acts\textemdash which yield objective information upon the resolution
of uncertainty\textemdash to model this consumption\textendash information
trade-off. We introduce several analytical tools to study information
incentives, including a novel notion of an information premium (analogous
to a risk premium) and a two-parameter ADP parameterization that independently
captures risk aversion and information incentives.

From an alternative perspective, the model can be interpreted as a
framework for information acquisition with a consumption cost. This
forgone consumption utility provides a microfoundation for why information
is inherently costly. In ongoing work, we explore an extension of
this model where\textemdash because consumption derives from expected
utility and beliefs converge more slowly toward the extremes\textemdash pairing
a linear consumption loss with diminishing information gains naturally
gives rise to convex information costs.

\begin{singlespace}\bibliographystyle{chicago}
\phantomsection\addcontentsline{toc}{section}{\refname}\bibliography{database_all}
\end{singlespace}

\newpage{}

\appendix

\section{\label{appx:Appendix:-Proofs}Appendix: Proofs}

\subsection{Proof of Theorem \ref{thm:IC-1}}

We prove ``only if'' as the other direction is straightforward.

\subsubsection*{Step 0: \label{subsec:notations-objects}Notation and objects}

For any $f\in\mathcal{F}$, we say that $f_{x}^{*}\in\mathcal{F}$
is a\emph{ spitting of prize} $x\in\text{supp}\left(f\right)$ if
$f_{x}^{*}$ can be obtained through the following operation: There
is a set of lotteries $\left\{ p_{1},...,p_{\left\{ \Omega\right\} }\right\} $
satisfying $\text{supp}\left(p_{i}\right)\ensuremath{\cap}\text{supp}\left(f\right)=\ensuremath{\emptyset}$
for all $i$, $\text{supp}\left(p_{i}\right)\ensuremath{\cap}\text{supp}\left(p_{j}\right)=\ensuremath{\emptyset}$
for all $i,j$, and $p_{i}\sim x$ for all $i$, such that 
\[
f_{x}^{*}\left(z\mid\omega_{i}\right)\coloneqq\begin{cases}
f\left(x\mid\omega_{i}\right)p_{i}\left(z\right) & z\in\text{supp}\left(p_{i}\right)\\
f\left(z\mid\omega_{i}\right) & z\ne x
\end{cases}.
\]
In words, $f_{x}^{*}$ is obtained by assigning a non-overlapping
state-dependent lottery replacement of prize $x$ in $f$. Given $f\in\mathcal{F}$
and $x\in X$, Let 
\[
F^{*}\left(f,x\right)
\]
be the set of all $f_{x}^{*}$. Note that $F^{*}\left(f,x\right)\ne\emptyset$
because $\text{supp}\left(f\right)$ is finite and $X=\mathbb{R}$.
Moreover, note that if $f_{x}$ is such that for some $\omega_{i}\in\Omega$
we have $f_{x}\left(\omega_{i}\right)>0$ and $f_{x}\left(\omega_{j}\right)=0$
for all $\omega_{j}\in\Omega\backslash\left\{ \omega_{i}\right\} $,
then $f\in F^{*}\left(f,x\right)$. For any $x$, $f_{x}^{*}\in F^{*}\left(f,x\right)$,
and $f_{x}^{*}{}'\in F^{*}\left(f',x\right)$, we say $\left(f_{x}^{*},f_{x}^{*}\right)$
\emph{involve common splits} if their operations involve the same
$\left\{ p_{1},...,p_{N}\right\} $. If $x\notin\text{supp}\left(f\right)$,
then let $F^{*}\left(f,x\right)=\left\{ f\right\} $.

For any $f\in\mathcal{F}$, define $f_{I}^{*}\in\mathcal{F}$ as follows:
Turn $\text{supp}\left(f\right)$ into an ordered set $\left\{ x_{1},...,x_{M}\right\} \subset X$.
For the first prize $x_{1}$, pick $f^{1}$ from $F^{*}\left(f,x_{1}\right)$.
If $M>1$, then for the second prize $x_{2}$, pick $f^{2}$ from
$F^{*}\left(f^{1},x_{2}\right)$ satisfying $\text{supp}\left(f^{2}\right)\cap\left\{ x_{1}\right\} =\emptyset$.
If $M>2$, iterate this process for each other $i$ to obtain $f^{3},...,f^{M}$,
each satisfying $\text{supp}\left(f^{i}\right)\cap\left\{ x_{1},...,x_{i-1}\right\} =\emptyset$,
which is possible at every step because $\text{supp}\left(f^{i}\right)$
is finite and $X=\mathbb{R}$. Finally, let $f_{I}^{*}\coloneqq f^{M}$.
In general, $f_{I}^{*}$ is not unique (given $f$) because the process
of picking $f^{1},f^{2},...$ involve freedom. Let 
\[
I^{*}\left(f\right)
\]
denote the set of all $f_{I}^{*}$. Note that $I^{*}\left(f\right)\ne\emptyset$
because $\text{supp}\left(f\right)$ is finite and $X=\mathbb{R}$. 

If an act $f\in\mathcal{F}$ is such that for any $x\in\text{supp}\left(I\right)$,
there exists $\omega_{i}\in\Omega$ such that $f\left(x\mid\omega_{i}\right)>0$
and $f\left(x\mid\omega_{j}\right)=0$ for all $\omega_{j}\ne\omega_{i}$,
then we call it a (fully) \emph{informative act}. Let $\mathcal{F}_{I}$
denote the set of all informative acts, and its typical element is
denoted by $I$. Clearly, $I^{*}\left(f\right)\subset\mathcal{F}_{I}$
for all $f$.

If an act $f\in\mathcal{F}$ is such that for all $\omega_{i}\in\Omega$,
$f\left(x_{i}\mid\omega_{i}\right)=1$ for some $x_{i}\in X$, and
$x_{i}\ne x_{j}$ for all $i\ne j$, then we call it a \emph{simple
act}. Let $\mathcal{F}_{S}$ denote the set of all simple acts, and
its typical element is denoted by $f^{s}$. Clearly, $\mathcal{F}_{S}\subset\mathcal{F}_{I}$.

Let $\mathcal{F}_{P}\subset\mathcal{F}$ denote the set of constant
acts. Let $\mathcal{F}_{X}\subset\mathcal{F}$ denote the set of degenerate
constant acts.

\subsubsection*{Step 1: Construction of $U$ and its properties}
\begin{lem}
\label{lem:Monotonicity}If $f\succ g$ and $f_{x}\propto g_{x}$
for all $x\in\text{supp}\left(f\right)$, then $af+\left(1-a\right)g\succ\beta f+\left(1-\beta\right)g$
for all $a,\beta\in\left[0,1\right]$ such that $a>\beta$.
\end{lem}
\begin{proof}
Fix qualifying $f,g,a,\beta$. By \axmref{non-overlapping}, $kf+\left(1-k\right)g\succ kg+\left(1-k\right)g$
for $k=\frac{a-\beta}{1-\beta}$, then since $\left(kf+\left(1-k\right)g\right)_{x}\propto f_{x}$
for any $x\in\text{supp}\left(f\right)$, by \axmref{non-overlapping}
again, $af+\left(1-a\right)g=\beta f+\left(1-\beta\right)\left(kf+\left(1-k\right)g\right)\succ\beta f+\left(1-\beta\right)g$. 
\end{proof}
\begin{claim}
\label{claim:betweeness}For any $f\in\mathcal{F}$, there exist a
degenerate prize $\delta_{x}\in\mathcal{F}$ and a simple act $g\in\mathcal{F_{S}}$
such that $g\succsim f\succsim\delta_{x}$.
\end{claim}
\begin{proof}
Fix $f\in\mathcal{F}$. We first show there exists a degenerate prize
$\delta_{x}\in\mathcal{F}$ such that $f\succsim\delta_{x}$. Let
$x$ be $\min_{\succeq}\text{supp}\left(f\right)$. We iteratively
replacing other prizes in the support of $f$ with prize $x$ using
the following process. Pick $y_{1}\in\text{supp}\left(f\right)\setminus\left\{ x\right\} $.
Let $f^{1}=fy_{1}x$. Therefore, by prize independence, since $y\succsim x$,
we have $f\succsim f^{1}$. Pick $y_{2}\in\text{supp}\left(f^{1}\right)\setminus\left\{ x\right\} $
and let $f^{2}=f^{1}y_{2}x$. Since $y_{2}\succsim x$, then by prize
independence we have $f^{1}\succsim f^{2}$. By iteratively generating
$f^{i}$, we finally have $f^{N}=\delta_{x}$, where $N=|\text{supp\ensuremath{\left(f\right)}}|-1$.
Then $f\succsim f^{1}\succsim f^{2}\succsim...\succsim f^{N}$. Therefore
$f\succsim\delta_{x}$.

We next show there exist a simple act $g\in\mathcal{F_{S}}$ such
that $g\succsim f$. Let $\mathcal{F}\left(f\right)=\left\{ h\in\mathcal{F}:\forall\omega\in\Omega,\exists x\in\text{supp}\left(f\left(\omega\right)\right)s.t.\ h\left(x|\omega\right)=1\right\} $.
Then there exists a convex combination over $\mathcal{F}\left(f\right)$
equivalent to $f$. Denote this convex combination by $\sum_{f^{i}\in\mathcal{F}\left(f\right)}a_{i}f^{i}$,
where $a_{i}\geq0,\sum a_{i}=1$. For every $i\in|\mathcal{F}\left(f\right)|$,
pick a simple act$f_{s}^{i}$ such that $f_{s}^{i}\cap f_{s}^{j}=\emptyset$
for any $j\neq i$ and for any $x\in\text{supp}\left(f_{s}^{i}\right)$,
$\delta_{x}\succsim\delta_{y}$ for all $y\in\text{supp}\left(f\right)$.
We can do the combination of simple acts because the prize space is
a real line. Denote $\sum_{i\in|\mathcal{F}\left(f\right)|}a_{i}f_{s}^{i}$
as $h$. Notice that $h$ is an informative act. We can obtain $\sum_{f^{i}\in\mathcal{F}\left(f\right)}a_{i}f^{i}$
from $h$ by the following process. For every $i\in|\mathcal{F}\left(f\right)|$,
replace $f_{s}^{i}\left(\omega\right)$ with $f^{i}\left(\omega\right)$
for any $\omega\in\Omega$. And since $f_{s}^{i}\left(\omega\right)\succsim f^{i}\left(\omega\right)$
for any $\omega\in\Omega$, by prize independence we have $h\succsim\sum_{f^{i}\in\mathcal{F}\left(f\right)}a_{i}f^{i}$,
i.e., $h\succsim f$. Then, there exists a simple act $g\in\mathcal{F}_{S}$
such that for any $x\in\text{supp}\text{\ensuremath{\left(g\right)}}$,
$x\succsim y$ for any $y\in\text{supp}\left(h\right)$. By \axmref{non-overlapping},
we have $g\left(\omega\right)\succsim h\left(\omega\right)$ for any
$\omega\in\Omega$. Then by iteratively replacing $g\left(\omega_{i}\right)$
with $h\left(\omega_{i}\right)$ for every $i\in|\Omega|$, we can
obtain $h$ from $g$. Therefore, by \axmref{prize_independence},
since $g\left(\omega\right)\succsim h\left(\omega\right)$ for any
$\omega\in\Omega$, we have $g\succsim h$. Therefore, $g\succsim h\succsim f$.
\end{proof}
\begin{lem}
\label{lem:U-and-affine-on-prop}There exists a utility function $U:\mathcal{F}\rightarrow\mathbb{R}$
such that \textup{the following hold:}
\begin{enumerate}
\item \textup{For any $f,g\in\mathcal{F}$,} $f\succsim g$\textup{ if and
only if} $U\left(f\right)\geq U\left(g\right)$.
\item \textup{For any $f,g\in\mathcal{F}$} \textup{such that $f_{x}\propto g_{x}$}
\textup{for any }$\ensuremath{x\in}\text{supp}\ensuremath{\left(f\right)}$,
$U\left(af+\left(1-a\right)g\right)=aU\left(f\right)+\left(1-a\right)U\left(g\right)$\textup{
for all $a\in\left[0,1\right]$.}
\item If $\text{supp}\ensuremath{\left(f\right)}\cap\text{supp}\ensuremath{\left(g\right)=\emptyset}$,
then $U\left(af+\left(1-a\right)g\right)=aU\left(f\right)+\left(1-a\right)U\left(g\right)$\textup{
for all $a\in\left[0,1\right]$.}
\end{enumerate}
\end{lem}
\begin{proof}
Notice that \axmref{non-overlapping} includes the standard independence
axiom on constant acts, that is, $p\succsim q\Leftrightarrow ap+\left(1-a\right)r\succsim aq+\left(1-a\right)r$
for all constant acts $p,q,r\in\mathcal{F}_{P}$ and $a\in\left[0,1\right]$.
Fix constant, degenerate acts $\delta_{b},\delta_{0}\in\mathcal{F}$
such that $\delta_{b}\succ\delta_{0}$, there exists a utility function
$\bar{U}:\Delta\left(X\right)\rightarrow\mathbb{R}$ such that $\bar{U}\left(\delta_{0}\right)=0$,
$\bar{U}\left(\delta_{b}\right)=1$, and for any $p,q\in\mathcal{F}_{P}$,
$\bar{U}\left(ap+\left(1-a\right)q\right)=a\bar{U}\left(p\right)+\left(1-a\right)\bar{U}\left(q\right)$.

Define $f\propto g$ if $\ensuremath{f_{x}\propto g_{x}}$ for any
$\ensuremath{x\in}\text{supp}\ensuremath{\left(f\right)}$. Let
\[
\mathcal{F}_{b,0}:=\left\{ f:\delta_{b}\succsim f\succsim\delta_{0},f\propto\delta_{b},f\propto\delta_{0}\right\} .
\]
For every $f\in\mathcal{F}_{b,0}$, Continuity in \axmref{standard}
and \lemref{Monotonicity} together guarantee a unique number $\lambda_{f}\in\left[0,1\right]$
such that $\lambda_{f}\delta_{b}+\left(1-\lambda_{f}\right)\delta_{0}\sim f$;
set $U_{0}\left(f\right)=\lambda_{f}$. Note that for all $f,g\in\mathcal{F}_{b,0}$,
Weak Order in \axmref{standard} gives $f\succsim g\Leftrightarrow U_{0}\left(f\right)\geq U_{0}\left(g\right)$.
Note also that $U_{0}\left(\delta_{0}\right)=\bar{U}\left(\delta_{0}\right)=0$
and $U_{0}\left(\delta_{b}\right)=\bar{U}\left(\delta_{b}\right)=1$,
so $U_{0}\left(p\right)=\bar{U}\left(p\right)$ for any $p\in\mathcal{F}_{P}\cap\mathcal{F}_{b,0}$.
We have constructed $U_{0}:\mathcal{F}_{b,0}\rightarrow\mathbb{R}$.

Next, consider any $f,g\in\mathcal{F}_{b,0}$ such that $f\propto g$.
We show that for any $a\in\left[0,1\right]$, $U_{0}\left(af+\left(1-a\right)g\right)=aU_{0}\left(f\right)+\left(1-a\right)U_{0}\left(g\right)$.

Since $f,g\in\mathcal{F}_{b,0}$ implies $f,g\succsim\delta_{0}$,
\axmref{non-overlapping} gives $af+\left(1-a\right)g\succsim a\delta_{0}+\left(1-a\right)g$
and $a\delta_{0}+\left(1-a\right)g\succsim a\delta_{0}+\left(1-a\right)\delta_{0}$,
then Weak Order in \axmref{standard} gives $af+\left(1-a\right)g\succsim\delta_{0}$.
Similarly, $\delta_{b}\succsim f,g$ leads to $\delta_{b}\succsim af+\left(1-a\right)g$.
So $af+\left(1-a\right)g\in\mathcal{F}_{b,0}$.

Fix any $f,g\in\mathcal{F}_{b,0}$. By $U_{0}$'s definition, $U_{0}\left(f\right)\delta_{b}+\left(1-U_{0}\left(f\right)\right)\delta_{0}\sim f$
and $U_{0}\left(g\right)\delta_{b}+\left(1-U_{0}\left(g\right)\right)\delta_{0}\sim g$.
Since we have $f_{z}\propto\left(\delta_{0}\right)_{z}$, $f_{z}\propto p_{z}$,
and $f_{z}\propto\left(\delta_{0}\right)_{z}$ for all $z\in\text{supp}\left(f\right)$,
multiple applications of \axmref{non-overlapping} give 
\begin{align*}
 & af+\left(1-a\right)g\\
 & \sim af+\left(1-a\right)\left(U_{0}\left(g\right)\delta_{b}+\left(1-U_{0}\left(g\right)\right)\delta_{0}\right)\\
 & \sim a\left[U_{0}\left(f\right)\delta_{b}+\left(1-U_{0}\left(f\right)\right)\delta_{0}\right]+\left(1-a\right)\left[U_{0}\left(g\right)\delta_{b}+\left(1-U_{0}\left(g\right)\right)\delta_{0}\right]
\end{align*}
for all $a\in\left[0,1\right]$. Note that the last line is constant
act, on which the affine property of $U_{0}$ has already been established
(when $\bar{U}$ was defined), so
\begin{align*}
 & U_{0}\left(af+\left(1-a\right)g\right)\\
 & =U_{0}\left(a\left[U_{0}\left(f\right)\delta_{b}+\left(1-U_{0}\left(f\right)\right)\delta_{0}\right]+\left(1-a\right)\left[U_{0}\left(g\right)\delta_{b}+\left(1-U_{0}\left(g\right)\right)\delta_{0}\right]\right)\\
 & =aU_{0}\left(U_{0}\left(f\right)\delta_{b}+\left(1-U_{0}\left(f\right)\right)\delta_{0}\right)+\left(1-a\right)U_{0}\left(U_{0}\left(g\right)\delta_{b}+\left(1-U_{0}\left(g\right)\right)\delta_{0}\right)\\
 & =aU_{0}\left(f\right)+\left(1-a\right)U_{0}\left(g\right)
\end{align*}
for all $a\in\left[0,1\right]$.

We have established the two properties in \lemref{U-and-affine-on-prop}
on $U_{0}:\mathcal{F}_{b,0}\rightarrow\mathbb{R}$. Next we extend
these properties to all of $\mathcal{F}$. Let
\[
\mathcal{F}_{f_{i},g_{i}}\coloneqq\left\{ h\in\mathcal{F}:f_{i}\succ h\succ g_{i},f_{i}\propto h,g_{i}\propto h\right\} .
\]
For every $i\in\mathbb{\mathbb{N}}_{+}$, pick $f_{i}\in\mathcal{F}_{S}$
and $g_{i}\in\mathcal{F}_{X}$ such that the following conditions
hold: (1) $\left\{ \delta_{b},\delta_{0}\right\} \cap\left(\text{supp}\left(f_{i}\right)\cup\text{supp}\left(g_{i}\right)\right)=\emptyset$
for all $i\in\mathbb{\mathbb{N}}_{+}$, (2) $f_{i}\succ f_{i-1}\succ g_{i-1}\succ g_{i}$
for all $i\in\mathbb{\mathbb{N}}_{+}$. By \claimref{betweeness},
for any $h\in\mathcal{F}$, there exists $k\in\mathbb{\mathbb{N}}$
such that $h\in\mathcal{F}_{f_{k},g_{k}}$.

Fix any $i\in\mathbb{N}_{+}$. For every $h\in\mathcal{F}_{f_{i},g_{i}}$,
Continuity in \axmref{standard} and \lemref{Monotonicity} together
guarantee a unique number $\lambda_{h}\in\left[0,1\right]$ such that
$\lambda_{f}f_{i}+\left(1-\lambda_{f}\right)g_{i}\sim h$; set $U_{i}^{*}\left(h\right)=\lambda_{h}$.
Note that using the same line of arguments as before, $U_{i}^{*}\left(f\right)\geq U_{i}^{*}\left(g\right)\Leftrightarrow f\succsim g$
for all $f,g\in\mathcal{F}_{f_{i},g_{i}}$, and $U_{i}\left(af+\left(1-a\right)g\right)=aU_{i}\left(f\right)+\left(1-a\right)U_{i}\left(g\right)$
for all $f,g\in\mathcal{F}_{f_{i},g_{i}}$ and $a\in\left[0,1\right]$.
Let $U_{i}$ be a positive affine transformation of $U_{i}^{*}$ so
that $U_{i}\left(\delta_{b}\right)=1$ and $U_{i}\left(\delta_{0}\right)=0$.
For later, we observe that for any degenerate prize $\delta_{x}\in\mathcal{F}_{f_{i},g_{i}}$,
we have $U_{i}\left(\delta_{x}\right)=\bar{U}\left(\delta_{x}\right)$
for every $i\in\mathbb{N}_{+}$. 

Fix any $h\in\mathcal{F}_{f_{i},g_{i}}\cap\mathcal{F}_{f_{j},g_{j}}$.
Fix $x,y\in X$ such that $\delta_{x}\succ\delta_{y}$, $\left\{ \delta_{x},\delta_{y}\right\} \subset\left(\mathcal{F}_{f_{i},g_{i}}\cap\mathcal{F}_{f_{j},g_{j}}\right)$,
and $\text{supp}\left(h\right)\cap\left\{ x,y\right\} =\emptyset$;
their existence is guaranteed because $\mathcal{F}_{\delta_{b},\delta_{0}}\subset\mathcal{F}_{f_{i},g_{i}}$
and $\mathcal{F}_{\delta_{b},\delta_{0}}$ contains uncountably many
elements of $X$. By \lemref{Monotonicity} and Continuity in \axmref{standard},
there are three possible cases:
\begin{align*}
\delta_{x} & \succsim h\succsim\delta_{y},h\sim a\delta_{x}+\left(1-a\right)\delta_{y},\\
\delta_{x} & \succ\delta_{y}\succsim h,\delta_{y}\sim a\delta_{x}+\left(1-a\right)h,\\
h & \succsim\delta_{x}\succ\delta_{y},\delta_{x}\sim ah+\left(1-a\right)\delta_{y},
\end{align*}
where $a\in\left[0,1\right]$ is unique. Since we have established
that for all $k\in\mathbb{N}$, $U_{k}$ represent preference on $\mathcal{F}_{f_{k},g_{k}}$
and $U_{k}$ has affine structure for any $f,g\in\mathcal{F}_{f_{k},g_{k}}$
such that $f\propto g$, then, for each the above possibilities, we
have for $k\in\left\{ i,j\right\} $,
\begin{align*}
U_{k}\left(f\right) & =aU_{k}\left(\delta_{x}\right)+\left(1-a\right)U_{k}\left(\delta_{y}\right),\\
U_{k}\left(\delta_{y}\right) & =aU_{k}\left(\delta_{x}\right)+\left(1-a\right)U_{k}\left(h\right),\\
U_{k}\left(\delta_{x}\right) & =aU_{k}\left(h\right)+\left(1-a\right)U_{k}\left(\delta_{y}\right),
\end{align*}
respectively. Since $U_{k}\left(\delta_{x}\right)=U_{0}\left(\delta_{x}\right)$
and $U_{k}\left(\delta_{y}\right)=U_{0}\left(\delta_{y}\right)$ for
all $k\in\mathbb{N}$, so $U_{i}\left(h\right)=U_{j}\left(h\right)$
regardless of which of the three possible cases is true.

Finally, let $U:\mathcal{F}\rightarrow\mathbb{R}$ be the common value
of $U_{i}$. Assured by the above arguments, $U$ is well-defined.
Fix any $f,g\in\mathcal{F}$, note that there exists $\mathcal{F}_{f_{i},g_{i}}$
such that $f,g\in\mathcal{F}_{f_{i},g_{i}}$. Note that 
\[
U\left(f\right)=U_{i}\left(f\right)\geq U_{i}\left(g\right)=U\left(g\right)\Leftrightarrow f\succsim g,
\]
so have shown property 1 in \lemref{U-and-affine-on-prop}. Moreover,
if $f_{x}\propto g_{x}$ for all $x\in\text{supp}\left(f\right)$,
then 
\[
U\left(af+\left(1-a\right)g\right)=U_{i}\left(af+\left(1-a\right)g\right)=aU_{i}\left(f\right)+\left(1-a\right)U_{i}\left(g\right)=aU\left(f\right)+\left(1-a\right)U\left(g\right)
\]
for all $a\in\left[0,1\right]$, so we have shown property 2 in \lemref{U-and-affine-on-prop}.
Property 3 in \lemref{U-and-affine-on-prop} is simply because $\text{supp}\ensuremath{\left(f\right)}\cap\text{supp}\ensuremath{\left(g\right)=\emptyset}$
implies $f_{x}\propto g_{x}$ for any $\ensuremath{x\in}\text{supp}\ensuremath{\left(f\right)}$.
\end{proof}
\begin{claim}
\label{claim:Splitting indifference}For any $f\in\mathcal{F}$ and
$f_{x}^{*},f_{x}^{*}{}'\in F^{*}\left(f,x\right)$, $U\left(f_{x}^{*}\right)=U\left(f_{x}^{*}{}'\right)$.
\end{claim}
\begin{proof}
This is straightforward by multiple applications of \axmref{prize_independence}.
\end{proof}
\begin{claim}
\label{claim:I-splitting-indifference}For any $f\in\mathcal{F}$
and $f_{I}^{*},f_{I}^{*}{}'\in I^{*}\left(f\right)$, $U\left(f_{I}^{*}\right)=U\left(f_{I}^{*}{}'\right)$.
\end{claim}
\begin{proof}
By applying \claimref{{Splitting indifference}} iteratively.
\end{proof}

\subsubsection*{Step 2: Construction of $v$ bar and its properties}

Let the function $\bar{v}:\left[0,1\right]^{N}\rightarrow\mathbb{R}$
satisfy
\[
\bar{v}\left(f_{x}\right)=U\left(f_{x}^{*}\right)-U\text{\ensuremath{\left(f\right)}}
\]
where $f_{x}^{*}\in F^{*}\left(f,x\right)$ (defined in \subsecref{notations-objects}).
Note that $\bar{v}\left(\boldsymbol{0}\right)=0$ because $f_{x}=\boldsymbol{0}$
implies $x\notin\text{supp}\left(f\right)$, which implies $F^{*}\left(f,x\right)=\left\{ f\right\} $
by definition. Moreover, note that if $f_{x}$ is such that for some
$\omega_{i}\in\Omega$, $f_{x}\left(\omega_{i}\right)>0$ and $f_{x}\left(\omega_{j}\right)=0$
for all $\omega_{j}\in\Omega\backslash\left\{ \omega_{i}\right\} $,
then $f\in F^{*}\left(f,x\right)$, so $\bar{v}\left(f_{x}\right)=0$.
\begin{claim}
\label{claim:well-defined}$\bar{v}$ is well-defined.
\end{claim}
\begin{proof}
For any $f\in\mathcal{\mathcal{F}}$ and $x\in\text{supp\ensuremath{\left(f\right)}}$,
let
\begin{align*}
\mathcal{F}_{1}\left(f,x\right) & \coloneqq\left\{ h\in\mathcal{F}:h_{x}=f_{x},h_{y}=\boldsymbol{1}_{N}-f_{x},\exists y\neq x\right\} \\
\mathcal{F}_{2}\left(f,x\right) & \coloneqq\left\{ h\in\mathcal{F}:h_{z}=f_{x},h_{z'}=\boldsymbol{1}_{N}-f_{x},\exists z,z'\in X\right\} .
\end{align*}
To show that $\bar{v}\left(f_{x}\right)$ is well-defined, we need
to establish two statements.
\begin{enumerate}
\item For any $f\in\mathcal{F}$ and $x\in\text{supp\ensuremath{\left(f\right)}}$,
if $h\in\mathcal{F}_{1}\left(f,x\right)$, then $U\left(f_{x}^{*}\right)-U\text{\ensuremath{\left(f\right)}}=U\left(h_{x}^{*}\right)-U\text{\ensuremath{\left(h\right)}}$.
\item For any $f\in\mathcal{F}$, $x\in\text{supp\ensuremath{\left(f\right)}}$,
and $h\in\mathcal{F}_{2}\left(f,x\right)$, if $h_{z}=f_{x}$, then
$U\left(f_{x}^{*}\right)-U\text{\ensuremath{\left(f\right)}}=U\left(h_{z}^{*}\right)-U\text{\ensuremath{\left(h\right)}}$.
\end{enumerate}
\textbf{To prove statement 1}, fix any $f\in\mathcal{F}$ and $x\in\text{supp\ensuremath{\left(f\right)}}$.
Consider any $h\in\mathcal{F}_{1}\left(f,x\right)$ such that $\text{supp\ensuremath{\left(f\right)}\ensuremath{\cap}}\text{supp\ensuremath{\left(h\right)=\left\{ x\right\} }}$.
Fix any $h_{x}^{*}\in F^{*}\left(h,x\right)$ such that $\text{supp\ensuremath{\left(f\right)}\ensuremath{\cap}}\text{supp\ensuremath{\left(h_{x}^{*}\right)=\left\{ x\right\} }}$
and $\left(f_{x}^{*},h_{x}^{*}\right)$ involve common splits (\subsecref{notations-objects}).
By definition of split, $\frac{1}{2}f_{x}^{*}+\frac{1}{2}h=\frac{1}{2}h_{x}^{*}+\frac{1}{2}f$.
Since $\left(f_{x}^{*}\right)_{z}\propto h_{z}$ for all $z\in\text{supp\ensuremath{\left(h\right)}}$
and $\left(h_{x}^{*}\right)_{z}\propto f_{z}$ for all $z\in\text{supp\ensuremath{\left(f\right)}}$,
so \lemref{U-and-affine-on-prop} gives the first and last equalities
in
\[
\frac{1}{2}U\left(f_{x}^{*}\right)+\frac{1}{2}U\left(h\right)=U\left(\frac{1}{2}f_{x}^{*}+\frac{1}{2}h\right)=U\left(\frac{1}{2}h_{x}^{*}+\frac{1}{2}f\right)=\frac{1}{2}U\left(h_{x}^{*}\right)+\frac{1}{2}U\left(f\right).
\]
So we have $U\left(f_{x}^{*}\right)-U\text{\ensuremath{\left(f\right)}}=U\left(h_{x}^{*}\right)-U\text{\ensuremath{\left(h\right)}}$.
Fix one such $h$, which exists because $\text{supp}\left(f\right)$
is finite and $X=\mathbb{R}$. Consider any $h'\in\mathcal{F}_{1}\left(f,x\right)$
such that $\text{supp\ensuremath{\left(f\right)}\ensuremath{\cap}}\text{supp\ensuremath{\left(h'\right)=\left\{ x,z\right\} }}$
for some $z\in X$ and $\text{supp\ensuremath{\left(h\right)}\ensuremath{\cap}}\text{supp\ensuremath{\left(h'\right)=\left\{ x\right\} }}$.
Consider $h_{x}^{*}\in F^{*}\left(h,x\right)$ and $h'{}_{x}^{*}\in F^{*}\left(h',x\right)$
such that $\left(h_{x}^{*},h'{}_{x}^{*}\right)$ involve common splits.
Then we have $\frac{1}{2}h'{}_{x}^{*}+\frac{1}{2}h=\frac{1}{2}h_{x}^{*}+\frac{1}{2}h'$,
and therefore 
\[
U\left(f_{x}^{*}\right)-U\text{\ensuremath{\left(f\right)}}=U\left(h_{x}^{*}\right)-U\text{\ensuremath{\left(h\right)}}=U\left(h'{}_{x}^{*}\right)-U\text{\ensuremath{\left(h'\right)}}
\]
 using a similar argument as before. Note that every act in $\mathcal{F}_{1}\left(f,x\right)$
has a support of at most size two, so we have covered all of $\mathcal{F}_{1}\left(f,x\right)$. 

\textbf{To prove statement 2}, fix any $f_{0}\in\mathcal{F}$ such
that $x\in\text{supp}$$\left(f_{0}\right)$. Pick $f\in\mathcal{F}_{1}\left(f_{0},x\right)$
such that $f_{x}=\left(f_{0}\right)_{x}$ and $\text{supp}\left(f\right)=\left\{ x,z\right\} $.
Consider any $h\in\mathcal{F}_{2}\left(f_{0},x\right)$ such that
$h_{y}=f_{x}$ for some $y\neq x$ and $h_{z}=f_{z}$. Let $\bar{h}\coloneqq\frac{1}{2}f+\frac{1}{2}h$,
then $\bar{h}_{x}^{*}=\frac{1}{2}f_{x}^{*}+\frac{1}{2}h$ and $\bar{h}_{y}^{*}=\frac{1}{2}h_{y}^{*}+\frac{1}{2}f$.
Consider the constant act $p_{1}\in\mathcal{F}_{P}$ such that $p_{1}\sim\delta_{x}$
and $\text{supp\ensuremath{\left(\bar{h}\right)\cap\text{supp\ensuremath{\left(p_{1}\right)}}=\emptyset}}$.
Let $\bar{h}_{x1}=\bar{h}_{\omega_{1}}xp_{1}$. Consider $p_{2}\in\mathcal{F}_{P}$
such that $p_{2}\sim\delta_{x}$ and $\text{supp\ensuremath{\left(\bar{h}_{x1}\right)\cap\text{supp\ensuremath{\left(p_{2}\right)}}=\emptyset}}$.
Let $p_{2}$ replace $x$ under $\omega_{1}$ in $\bar{h}_{x1}$,
denote the new act $\bar{h}_{x1}{}_{\omega_{1}}xp_{2}$ by $\bar{h}_{x2}$.
Iterate the process on $\bar{h}_{xi}$. Finally we get the ultimate
generated act $\bar{h}_{xN}$ using some $p_{1},...p_{n}$ in this
process, and $\bar{h}_{x}^{*}=\text{\ensuremath{\bar{h}_{xN}}}$.
Consider $q_{1}\in\mathcal{F}_{P}$ such that $q_{1}\sim\delta_{y}$
and $\text{supp\ensuremath{\left(\bar{h}\right)\cap\text{supp\ensuremath{\left(q_{1}\right)}}=\emptyset}}$.
Conduct the same process on $\bar{h}_{\omega_{1}}yq_{1}$. In the
end we get the ultimate generated act $\bar{h}_{yN}$; note that $\bar{h}_{y}^{*}=\text{\ensuremath{\bar{h}_{yN}}}$.
By iteratively applying \axmref{state_prize_independence} in the
process of generating $\bar{h}_{xN}$ and $\bar{h}_{yN}$ , we have
$\bar{h}_{x}^{*}\sim\bar{h}_{y}^{*}$. Therefore, $\frac{1}{2}f_{x}^{*}+\frac{1}{2}h\sim\frac{1}{2}h_{y}^{*}+\frac{1}{2}f$.
Since $f_{x_{z}}^{*}\propto h_{z}$ for any $z\in\text{supp}\left(h\right)$
and $f\propto h_{y_{z}}^{*}$ for any $z\in\text{supp}\left(f\right)$,
using \lemref{U-and-affine-on-prop}, we get $U\left(f_{x}^{*}\right)-U\text{\ensuremath{\left(f\right)}}=U\left(h_{y}^{*}\right)-U\text{\ensuremath{\left(h\right)}}$.
Consider any $h'\in\mathcal{F}_{2}\left(f_{0},x\right)$ such that
$\text{supp}\left(h'\right)=\left\{ y,z'\right\} $, $h'_{y}=h_{y}$,
and $h'_{z'}=h_{z}$. So
\begin{align*}
U\left(\left(f_{0}^{*}\right)_{x}\right)-U\text{\ensuremath{\left(f_{0}\right)}}=U\left(f_{x}^{*}\right)-U\text{\ensuremath{\left(f\right)}}=U\left(h_{y}^{*}\right)-U\text{\ensuremath{\left(h\right)}}=U\left(h'{}_{y}^{*}\right)-U\text{\ensuremath{\left(h'\right)}}
\end{align*}
where the first and last equalities are due to statement 1.

\end{proof}
\begin{claim}
\label{claim:summation of vbar}For any $f\in\mathcal{F}$ and $f_{I}^{*}\in I^{*}\left(f\right)$,
$\sum_{x}\bar{v}\left(f_{x}\right)=U\left(f_{I}^{*}\right)-U\left(f\right)$.
\end{claim}
\begin{proof}
For any $f\in\mathcal{F}$ and $f_{I}^{*}\in I^{*}\left(f\right)$,
record $\left\{ x_{1},...,x_{M}\right\} \subset X$ and $f^{1},f^{2},...,f^{M}$
from the definition (construction) of $f_{I}^{*}$. Let $f^{0}\coloneqq f$
and $f_{x_{1}}^{*}\coloneqq f^{1}$. If $M>1$, for every $i>1$,
pick a $f_{x_{i}}^{*}\in F^{*}\left(f,x_{i}\right)$ such that $\left(f_{x_{i}}^{*},f^{i}\right)$
involve common splits (\subsecref{notations-objects}). By definition,
$\frac{1}{2}f+\frac{1}{2}f^{i}=\frac{1}{2}f_{x_{i}}^{*}+\frac{1}{2}f^{i-1}$.
For any $i\in\left\{ 1,...,M\right\} $, since $f_{y}\propto f_{y}^{i}$
for any $y\in\text{supp}\left(f\right)$ and $\left(f_{x_{i}}^{*}\right)_{y}\propto f_{y}^{i-1}$
for any $y\in\text{supp}\left(f\right)$, \lemref{U-and-affine-on-prop}
gives the first and last equalities in
\[
\frac{1}{2}U\left(f\right)+\frac{1}{2}U\left(f^{i}\right)=U\left(\frac{1}{2}f+\frac{1}{2}f^{i}\right)=U\left(\frac{1}{2}f_{x_{i}}^{*}+\frac{1}{2}f^{i-1}\right)=\frac{1}{2}U\left(f_{x_{i}}^{*}\right)+\frac{1}{2}U\left(f^{i-1}\right).
\]
Rearranging gives $U\left(f_{x_{i}}^{*}\right)-U\left(f\right)=U\left(f^{i}\right)-U\left(f^{i-1}\right)$
for all $i\in\left\{ 1,...,M\right\} $. Finally,
\begin{align*}
\sum_{x\in\text{supp}\left(f\right)}\bar{v}\left(f_{x}\right) & =\sum_{x\in\text{supp}\left(f\right)}U\left(f_{x}^{*}\right)-U\text{\ensuremath{\left(f\right)}}\\
 & =\sum_{i=1,...,M}U\left(f^{i}\right)-U\left(f^{i-1}\right)=U\left(f^{M}\right)-U\left(f^{0}\right)\\
 & =U\left(f_{I}^{*}\right)-U\left(f\right)
\end{align*}
as desired.
\end{proof}
\begin{claim}
\label{claim:homogeneous in degree 1}For any $f_{x}\in\left[0,1\right]^{N}$,
if $kf_{x}\in\left[0,1\right]^{N}$for $k\geq0$, then $\bar{v}\left(kf_{x}\right)=k\bar{v}\left(f_{x}\right)$.
Also, $\bar{v}\left(\cdot\right)\geq0$.
\end{claim}
\begin{proof}
To show that $\bar{v}$ is homogeneous of degree 1, we focus on binary
state space ($N=2$) for ease of exposition. The proof works for any
finite state space. Fix $\left(a,b\right),\left(ka,kb\right)\in\left[0,1\right]^{N}$
where $k\geq0$. Consider the following acts:

\begin{tabular}{c|ccc|}
\multicolumn{1}{c}{$f$} & $x$ & $y$ & \multicolumn{1}{c}{$z$}\tabularnewline
\cline{2-4}
$\omega_{1}$ & $a$ & $1-a$ & $0$\tabularnewline
$\omega_{2}$ & $b$ & $0$ & $1-b$\tabularnewline
\cline{2-4}
\end{tabular}\hskip 2em%
\begin{tabular}{c|ccc|}
\multicolumn{1}{c}{$g$} & $x$ & $y$ & \multicolumn{1}{c}{$z$}\tabularnewline
\cline{2-4}
$\omega_{1}$ & $ka$ & $1-ka$ & $0$\tabularnewline
$\omega_{2}$ & $kb$ & $0$ & $1-kb$\tabularnewline
\cline{2-4}
\end{tabular}\hskip 2em%
\begin{tabular}{c|cc|}
\multicolumn{1}{c}{$h$} & $y$ & \multicolumn{1}{c}{$z$}\tabularnewline
\cline{2-3}
$\omega_{1}$ & $1$ & $0$\tabularnewline
$\omega_{2}$ & $0$ & $1$\tabularnewline
\cline{2-3}
\end{tabular}

\vskip 1emNote that for any $i\in\text{supp}\left(f\right)$, we
have $f_{i}\propto h_{i}$. Note also that $g=kf+\left(1-k\right)h$
and $g_{x}^{*}=kf_{x}^{*}+\left(1-k\right)h$. Then by \lemref{U-and-affine-on-prop},
$\bar{v}\left(g_{x}\right)=U\left(g_{x}^{*}\right)-U\left(g\right)$$=U\left(kf_{x}^{*}+\left(1-k\right)h\right)-U\left(g\right)$$=kU\left(f_{x}^{*}\right)+\left(1-k\right)U\left(h\right)-U\left(kf+\left(1-k\right)h\right)$$=k\bar{v}\left(f_{x}\right)$.
So $\bar{v}\left(kf_{x}\right)=k\bar{v}\left(f_{x}\right)$ as desired.

To show that $\bar{v}\left(\cdot\right)\geq0$, we need to show $f_{y}^{*}\succsim f$
for all $f\in\mathcal{F}$ and $f_{y}^{*}\in F^{*}\left(f,y\right)$.
Suppose $y\notin\text{supp}\left(f\right)$, then it is already established
that $\bar{v}\left(f_{y}\right)=0$. Suppose from now on $y\in\text{supp}\left(f\right)$.
We focus on binary state space ($N=2$) for ease of exposition. The
proof works for any finite state space. Suppose $ax+\left(1-a\right)z\sim a'x'+\left(1-a'\right)z'\sim a''x''+\left(1-a''\right)z''\sim y$
for distinct $x,z,x',z',x'',z''$ and $\left\{ x,z,x',z',x'',z''\right\} \cap\text{supp}\text{\ensuremath{\left(f\right)}}=\emptyset$.
Pick any $g_{1}\in F^{*}\left(f,y\right)$ where the split uses $ax+\left(1-a\right)z$
to replace $y$ under $\omega_{1}$ and uses $a'x'+\left(1-a'\right)z'$
replaces $y$ under $\omega_{2}$. Let $g_{2}\in F^{*}\left(f,y\right)$
where the split uses $a'x'+\left(1-a'\right)z'$ to replace $y$ under
$\omega_{1}$ and uses $a''x''+\left(1-a''\right)z''$ to replace
$y$ under $\omega_{2}$. Then by repeatedly applying \axmref{prize_independence},
we have $g_{1}\sim g_{2}$. Let $g_{3}\in F^{*}\left(f,y\right)$
where the split uses $a'x'+\left(1-a'\right)z'$ to replace $y$ under
$\omega_{1}$ and uses $ax+\left(1-a\right)z$ to replace $y$ under
$\omega_{2}$. Then by \axmref{prize_independence}, $g_{2}\sim g_{3}$,
then $g_{1}\sim g_{3}$. By \axmref{prize_independence}, $\frac{1}{2}g_{1}+\frac{1}{2}g_{3}\sim f$.
\axmref{informationseeking} implies $g_{1}\succsim\frac{1}{2}g_{1}+\frac{1}{2}g_{3}\sim f$.
Note that \claimref{{Splitting indifference}} gives $U\left(f_{y}^{*}\right)=U\left(g_{1}\right)$,
which gives $f_{y}^{*}\sim g_{1}$, and so $f_{y}^{*}\succsim f$
by weak order.

\end{proof}
\begin{claim}
\label{claim:concave}$\bar{v}$ is concave.
\end{claim}
\begin{proof}
Fix any $r,s\in\left[0,1\right]^{N}$, we need to show $\bar{v}\left(ar+\left(1-a\right)s\right)\geq a\bar{v}\left(r\right)+\left(1-a\right)\bar{v}\left(s\right)$.
Pick $f,g\in\mathcal{F}$ such that $f_{x}=r$, $g_{x}=s$, and $\text{supp\ensuremath{\left(f\right)}\ensuremath{\cap}supp\ensuremath{\left(g\right)}}=\left\{ x\right\} $.
Let $h\coloneqq af+\left(1-a\right)g$; clearly, 
\begin{equation}
h_{x}=af_{x}+\left(1-a\right)g_{x}.\label{eq:vbarconcave1}
\end{equation}
Pick $f_{I}^{*}\in I^{*}\left(f\right)$ and $g_{I}^{*}\in I^{*}\left(g\right)$
such that $\text{supp\ensuremath{\left(f_{I}^{*}\right)}}\cap\text{supp\ensuremath{\left(g_{I}^{*}\right)}}=\emptyset$,
so \lemref{U-and-affine-on-prop} gives the second equality in 
\begin{equation}
U\left(h_{I}^{*}\right)=U\left(af_{I}^{*}+\left(1-a\right)g_{I}^{*}\right)=aU\left(f_{I}^{*}\right)+\left(1-a\right)U\left(g_{I}^{*}\right),\label{eq:hi}
\end{equation}
where $h_{I}^{*}\in I^{*}\left(h\right)$ and the first equality is
due to $h_{I}^{*}=af_{I}^{*}+\left(1-a\right)g_{I}^{*}$. Observe
that
\begin{align}
U\left(af+\left(1-a\right)g\right) & \leq aU\left(f\right)+\left(1-a\right)U\left(g\right)\nonumber \\
\left[U\left(h_{I}^{*}\right)\right]-\left[U\left(af+\left(1-a\right)g\right)\right] & \geq\left[aU\left(f_{I}^{*}\right)+\left(1-a\right)U\left(g_{I}^{*}\right)\right]-\left[aU\left(f\right)+\left(1-a\right)U\left(g\right)\right]\nonumber \\
U\left(h_{I}^{*}\right)-U\left(h\right) & \geq a\left[U\left(f_{I}^{*}\right)-U\left(f\right)\right]+\left(1-a\right)\left[U\left(g_{I}^{*}\right)-U\left(g\right)\right]\nonumber \\
\sum_{z\in\text{supp}\left(h\right)}\bar{v}\left(h_{z}\right) & \geq a\sum_{z\in\text{supp}\left(f\right)}\bar{v}\left(f_{z}\right)+\left(1-a\right)\sum_{z'\in\text{supp}\left(g\right)}\bar{v}\left(g_{z'}\right)\label{eq:vbarconcavehahaha}\\
\bar{v}\left(h_{x}\right) & \geq a\bar{v}\left(f_{x}\right)+\left(1-a\right)\bar{v}\left(g_{x}\right),\label{eq:vbarconcave2}
\end{align}
where the first line is given by \axmref{informationseeking}; the
second line is by multiplying $\left(-1\right)$ and adding the LHS
and RHS of \eqref{hi}; the third line replaces $h=af+\left(1-a\right)g$
and rearranges; the fourth line is by \claimref{{summation of vbar}}.
To arrive at the last line, observe that for all $z\in\text{supp}\left(h\right)\backslash\left\{ x\right\} $,
either $z\in\text{supp}\left(f\right)$ or $\text{supp}\left(g\right)$
(and not both), so
\begin{align*}
\bar{v}\left(h_{z}\right) & =\bar{v}\left(af_{z}+\left(1-a\right)g_{z}\right)=\begin{cases}
\bar{v}\left(af_{z}\right)=a\bar{v}\left(f_{z}\right) & \text{if }z\in\text{supp}\left(f\right)\\
\bar{v}\left(\left(1-a\right)g_{z}\right)=\left(1-a\right)\bar{v}\left(g_{z}\right) & \text{if }z\in\text{supp}\left(g\right)
\end{cases},
\end{align*}
where \claimref{{homogeneous in degree 1}} is used, resulting in
cancellations on both sides of \eqref{vbarconcavehahaha} that yield
the last line. So \eqref{vbarconcave1} implies \eqref{vbarconcave2},
as desired.
\end{proof}
\begin{lem}
\label{lem:menu representation}There exist $v:\mathcal{F}\times\Omega\rightarrow\mathbb{R}\cup\left\{ -\infty\right\} $
and $\tilde{A}\subset\mathcal{F}$ such that for any $r\in\left[0,1\right]^{N}$,
\[
\bar{v}\left(r\right)=\sum_{\omega\in\Omega}r\left(\omega\right)\left[\sup_{g\in\tilde{A}}v\left(g,\omega\right)\right]-\sup_{g\in\tilde{A}}\left[\sum_{\omega\in\Omega}r\left(\omega\right)v\left(g,\omega\right)\right].
\]
\end{lem}
\begin{proof}
We have $\bar{v}$$\left(\cdot\right)\geq0$ due to \claimref{{homogeneous in degree 1}}and
concave due to \claimref{concave}. By \claimref{{homogeneous in degree 1}},
for any $r\in\left[0,1\right]^{N}$, we have $\bar{v}\left(r\right)=\left[\sum_{\omega\in\Omega}r\left(\omega\right)\right]\bar{v}\left(\tilde{r}\right)$,
where $\tilde{r}\left(\omega\right)=\frac{r\left(\omega\right)}{\sum_{\omega'\in\Omega}r\left(\omega'\right)}$.
Therefore, we can define $\tilde{v}:\Delta\left(\Omega\right)\rightarrow\mathbb{R}$
such that $\tilde{v}\left(\tilde{r}\right)=\bar{v}\left(\tilde{r}\right)$
for any $\tilde{r}\in\Delta\left(\Omega\right)$, then $\tilde{v}$
is concave and $\tilde{v}\geq0$. By the definition of $\bar{v}$,
$\bar{v}\left(q\right)=\tilde{v}\left(q\right)=0$ for degenerate
$q\in\Delta\left(\Omega\right)$, i.e., there exists one $\omega\in\Omega$
such that $q\left(\omega\right)=1$. By Proposition 15 in Annie Liang's
lecture notes entitled \emph{Information and Learning in Economic
Theory} (which references \citet{frankel2019quantifying}), there
exist a set $\tilde{A}\subset\mathcal{F}$ and a function $v:\tilde{A}\times\Omega\rightarrow\mathbb{R}\cup\left\{ -\infty\right\} $
such that $\tilde{v}\left(q\right)=\sum_{\omega\in\Omega}q\left(\omega\right)\left[\sup_{g\in\tilde{A}}v\left(g,\omega\right)\right]-\sup_{g\in\tilde{A}}\left[\sum_{\omega\in\Omega}q\left(\omega\right)v\left(g,\omega\right)\right]$.
Extend $v$'s domain to $\mathcal{F}$ by setting $v\left(f,\omega\right)=0$
for any $f\in\mathcal{F}\setminus\tilde{A}$ and any $\omega\in\Omega$.
Therefore, 
\begin{align*}
\bar{v}\left(r\right) & =\left[\sum_{\omega\in\Omega}r\left(\omega\right)\right]\tilde{v}\left(\tilde{r}\right)\\
 & =\sum_{\omega\in\Omega}r\left(\omega\right)\left[\sup_{g\in\tilde{A}}v\left(g,\omega\right)\right]-\sup_{g\in\tilde{A}}\left[\sum_{\omega\in\Omega}r\left(\omega\right)v\left(g,\omega\right)\right].
\end{align*}

\end{proof}

\subsubsection*{Step 3: Construction of $\tilde{U}_{I}$ and $U_{C}$ from $U$}

For all $f\in\mathcal{F}$, define
\begin{align}
\tilde{U}_{I}\left(f\right) & \coloneqq\sum_{x\in\text{supp}\left(f\right)}\sup_{\tilde{g}\in\tilde{A}}\sum_{\omega\in\Omega}f\left(x\mid\omega\right)v\left(\tilde{g},\omega\right)\label{eq:VI}\\
U_{C}\left(f\right) & \coloneqq U\left(f\right)-\tilde{U}_{I}\left(f\right).\label{eq:Uc-1}
\end{align}
 Let 
\[
\gamma_{1}\coloneqq\max_{\tilde{g}\in\tilde{A}}\sum_{\omega\in\Omega}v\left(\tilde{g},\omega\right)\,\,\,\,\,\&\,\,\,\,\,\gamma_{2}\coloneqq\sum_{\omega\in\Omega}\max_{\tilde{g}\in\tilde{A}}v\left(\tilde{g},\omega\right).
\]

\begin{claim}
\label{claim:gamma1gamma2}For any $p\in\mathcal{F}_{P}$, $\tilde{U}_{I}\left(p\right)=\gamma_{1}$.
For any $I\in\mathcal{F}_{I}$, $\tilde{U}_{I}\left(f\right)=\gamma_{2}$.
\end{claim}
\begin{proof}
For any constant act $p\in\mathcal{F}$, $\tilde{U}_{I}\left(p\right)=\sum_{x\in\text{supp}\left(f\right)}\max_{\tilde{g}\in\tilde{A}}\sum_{\omega\in\Omega}p\left(x\right)\left[v\left(\tilde{g},\omega\right)\right]=\sum_{x\in\text{supp}\left(f\right)}p\left(x\right)\max_{\tilde{g}\in\tilde{A}}\sum_{\omega\in\Omega}\left[v\left(\tilde{g},\omega\right)\right]=1\cdot\gamma_{1}$.
For any informative act $I\in\mathcal{F}_{I}$, for each $\omega\in\Omega$,
denote by $X_{\omega}\subset X$ the set of prizes that occur with
positive probability under $\omega$. Then $\tilde{U}_{I}\left(I\right)=\sum_{x\in\text{supp}\left(f\right)}\max_{\tilde{g}\in\tilde{A}}\sum_{\omega\in\Omega}I\left(x|\omega\right)\left[v\left(\tilde{g},\omega\right)\right]=\sum_{\omega\in\Omega}\sum_{x\in X_{\omega}}\max_{\tilde{g}\in\tilde{A}}I\left(x|\omega\right)\left[v\left(\tilde{g},\omega\right)\right]$$=\sum_{\omega\in\Omega}\sum_{x\in X_{\omega}}I\left(x|\omega\right)\max_{\tilde{g}\in\tilde{A}}\left[v\left(\tilde{g},\omega\right)\right]=\sum_{\omega\in\Omega}1\cdot\max_{\tilde{g}\in\tilde{A}}\left[v\left(\tilde{g},\omega\right)\right]=\gamma_{2}$.
\end{proof}
\begin{claim}
\label{claim:same_Uc}For any $p,q\in\mathcal{F}_{P}$, $p\sim q$
implies $U_{C}\left(p\right)=U_{C}\left(q\right)$. For any $f,g\in\mathcal{F}_{I}$,
$f\sim g$ implies $U_{C}\left(f\right)=U_{C}\left(g\right)$.
\end{claim}
\begin{proof}
The immediate consequences of \lemref{U-and-affine-on-prop}, \claimref{gamma1gamma2},
and \eqref{Uc-1}.
\end{proof}

\subsubsection*{Step 4: State-dependent utility $\bar{u}$}

Pick $N$ distinct prizes to form an enumerated set of benchmark prizes
$B\coloneqq\left\{ b_{1},...,b_{N}\right\} \subset X$. Let $\mathcal{F}_{B}\coloneqq\left\{ f\in\mathcal{F}:\text{supp}\left(f\right)\cap B=\emptyset\right\} $.
Consider the \emph{benchmark act} $f_{0}\in\mathcal{F}$ such that
$f_{0}\left(b_{i}\mid\omega_{i}\right)=1$ for all $i\in\left\{ 1,...,N\right\} $.
For any $x\in X\setminus B$ and $\omega_{i}\in\Omega$, denote by
$f_{x,\omega_{i}}$ the\emph{ statewise simple act} satisfying $f\left(x\mid\omega_{i}\right)=1$
and $f\left(b_{j}\mid\omega_{j}\right)=1$ for all $j\in\left\{ 1,...,N\right\} \backslash\left\{ i\right\} $.
For any $x\in X\setminus B$ and $\omega_{i}\in\Omega$, let
\begin{equation}
\bar{u}\left(x,\omega_{i}\right)\coloneqq U\left(f_{x,\omega_{i}}\right)-\frac{N-1}{N}U\left(f_{0}\right)-\frac{1}{N}\gamma_{2}.\label{eq:ubar}
\end{equation}

\begin{claim}
\label{claim:uc for simple act}For any simple act $f^{s}\in\mathcal{F}_{S}\cap\mathcal{F}_{B}$,
$U_{C}\left(f^{s}\right)=\sum_{\omega\in\Omega}\sum_{z\in\text{supp}\left(f\right)}f^{s}\left(z\mid\omega\right)\bar{u}\left(z,\omega\right)$.
\end{claim}
\begin{proof}
Take a qualifying $f^{s}$ and denote the prize it yields in state
$\omega_{i}$ (with probability 1) by $x_{i}$. Observe for later
that $\bar{u}\left(x_{i},\omega_{i}\right)=1\cdot\bar{u}\left(x_{i},\omega_{i}\right)=f^{s}\left(x_{i}\mid\omega_{i}\right)\bar{u}\left(x_{i},\omega_{i}\right)=\sum_{z\in\text{supp}\left(f\right)}f^{s}\left(z\mid\omega_{i}\right)\bar{u}\left(z,\omega_{i}\right)$
for all $i$. Observe that $\frac{1}{N}f^{s}+\frac{N-1}{N}f_{0}=\oplus_{i=\left\{ 1,...,N\right\} }\frac{1}{N}f_{x_{i},\omega_{i}}$.
Because $f^{s}\propto f_{0}$ and $\text{supp}\left(f_{x_{i},\omega_{i}}\right)\cap\text{supp}\left(f_{x_{j},\omega_{j}}\right)=\emptyset$
for all distinct $i,j$, \lemref{U-and-affine-on-prop} gives 
\begin{equation}
\frac{1}{N}U\left(f^{s}\right)+\frac{N-1}{N}U\left(f_{0}\right)=\frac{1}{N}\sum_{i\in\left\{ 1,...,N\right\} }U\left(f_{x_{i},\omega_{i}}\right).\label{eq:simple}
\end{equation}
\eqref{ubar} gives $\sum_{i=\left\{ 1,...,N\right\} }U\left(f_{x_{i},\omega_{i}}\right)=\sum_{i=\left\{ 1,...,N\right\} }\bar{u}\left(x_{i},\omega_{i}\right)+\left(N-1\right)U\left(f_{0}\right)+\gamma_{2}$,
which can be used to simplify \eqref{simple} into $U\left(f^{s}\right)-\tilde{U}_{I}\left(f^{s}\right)=\sum_{i\in\left\{ 1,...,N\right\} }\bar{u}\left(x_{i},\omega_{i}\right)=\sum_{\omega\in\Omega}\sum_{z\in\text{supp}\left(f\right)}f^{s}\left(z\mid\omega\right)\bar{u}\left(z,\omega\right)$,
where $\tilde{U}_{I}\left(f^{s}\right)=\gamma_{2}$ due to \claimref{gamma1gamma2}
and the second equality is due to an earlier observation in the proof.
Since $U_{C}\left(f^{s}\right)=U\left(f^{s}\right)-\tilde{U}_{I}\left(f^{s}\right)$
(by \eqref{Uc-1}), we are done.
\end{proof}
\begin{claim}
\label{claim:uc for informative act}For any informative act $I\in\mathcal{F}_{I}\cap\mathcal{F}_{B}$,
$U_{C}\left(I\right)=\sum_{\omega\in\Omega}\sum_{z\in\text{supp}\left(I\right)}I\left(z\mid\omega\right)\bar{u}\left(z,\omega\right)$.
\end{claim}
\begin{proof}
For any informative act $I\in\mathcal{F}_{I}\cap\mathcal{F}_{B}$,
define its corresponding simple act set
\[
S\left(I\right)\coloneqq\left\{ f^{s}\in\mathcal{F}_{S}\cap\mathcal{F}_{B}:f^{s}\left(\omega_{i}\right)\in\text{supp}\left(f\left(\omega_{i}\right)\right),\forall i\in\left\{ 1,...,N\right\} \right\} .
\]
Observe that $I$ is a convex combination of the acts in $S\left(I\right)$,
i.e., $I=\oplus_{f\in S\left(I\right)}a_{f}f$ where $a_{f}\in\left[0,1\right]$
for all $f\in S\left(I\right)$ and $\sum_{f\in S\left(I\right)}a_{f}=1$.
Since $f\propto g$ for any $f,g\in S\left(I\right)$, by \lemref{U-and-affine-on-prop},
we have $U\left(I\right)=\sum_{f\in S\left(I\right)}a_{f}U\left(f\right)$.
\claimref{gamma1gamma2} gives $\tilde{U}_{I}\left(f\right)=\tilde{U}_{I}\left(I\right)$
for all $f\in S\left(I\right)$, so $U_{C}\left(I\right)=\sum_{f\in S\left(I\right)}a_{f}U_{C}\left(f\right)$.
Finally,
\begin{align*}
U_{C}\left(I\right) & =\sum_{f\in S\left(I\right)}a_{f}U_{C}\left(f\right)=\sum_{f\in S\left(I\right)}\sum_{z\in\text{supp}\left(f\right)}\sum_{\omega\in\Omega}a_{f}f\left(z\mid\omega\right)\bar{u}\left(z,\omega\right)\\
 & =\sum_{z\in\text{supp}\left(I\right)}\sum_{\omega\in\Omega}I\left(z\mid\omega\right)\bar{u}\left(z,\omega\right),
\end{align*}
where the second equality is from \claimref{{uc for simple act}}
and the third equality is due to $I\left(x\mid\omega\right)=\sum_{f\in S\left(I\right)}a_{f}f\left(x\mid\omega\right)$
for any $x\in\text{supp}\left(f\right)$ and $\omega\in\Omega$.
\end{proof}

\subsubsection*{Step 5: State-independent utility $u$}
\begin{claim}
\label{claim:uc_seu_constant}There exists $u:X\rightarrow R$ such
that $U_{C}\left(p\right)=\sum_{z\in\text{supp}\left(p\right)}u\left(z\right)p\left(z\right)$
for all constant act $p\in\mathcal{F}_{P}$.
\end{claim}
\begin{proof}
By the definition of $\tilde{U}_{I}$, $\tilde{U}_{I}\left(p\right)=\sum_{\omega\in\Omega}\left[\max_{\tilde{g}\in\tilde{A}}u\left(\tilde{g},\omega\right)\right]$
for all $p\in\mathcal{F}_{P}$. By \lemref{U-and-affine-on-prop},
$U$ is affine on the set of all constant acts $\mathcal{F}_{P}$,
then $U_{C}\left(p\right)=U\left(p\right)-\tilde{U}_{I}\left(p\right)$
is affine on $\mathcal{F}_{P}$. It is well-known that there exists
$u:X\rightarrow R$ such that $U_{C}\left(p\right)=\sum_{z\in\text{supp}\left(p\right)}u\left(z\right)p\left(z\right)$
for all $p\in\mathcal{F}_{P}$. Note that $u\left(x\right)=U_{C}\left(\delta_{x}\right)$
for all $x\in X$.
\end{proof}
\begin{claim}
\label{claim:alpha-eta}For any $\omega\in\Omega$, there exist $a_{\omega}\geq0$
and $\eta_{\omega}\in\mathbb{R}$ such that $\bar{u}\left(z,\omega\right)=a_{\omega}u\left(z\right)+\eta_{\omega}$
for all $z\in X\setminus B$.
\end{claim}
\begin{proof}
Fix any $\omega\in\Omega$. Fix any $z\in X\setminus B$. Suppose
$\delta_{x}\succ\delta_{z}\succ\delta_{y}$ for some $x,y\in X\backslash B$.
By Continuity in \axmref{standard}, there exists $a\in\left(0,1\right)$
such that $a\delta_{x}+\left(1-a\right)\delta_{y}\sim\delta_{z}$,
let $q\coloneqq a\delta_{x}+\left(1-a\right)\delta_{y}$. Fix a simple
act $f^{s}\in\mathcal{F}_{S}\cap\mathcal{F}_{B}$ such that $f^{s}\left(\omega\right)=z$
and $\left\{ \delta_{x},\delta_{y}\right\} \cap\text{supp}\left(f^{s}\right)=\emptyset$.
The existence of $f^{s}$ is guaranteed by $\left|X\backslash B\right|=\infty$
and $\left|\Omega\right|<\infty$. and \axmref{prize_independence}
gives $f^{s}zq\sim f^{s}$, then \lemref{U-and-affine-on-prop} gives
$U\left(f^{s}zq\right)=U\left(f^{s}\right)$. Consider $f_{1}\in\mathcal{F}_{S}\cap\mathcal{F}_{B}$
where $f_{1}\left(\omega\right)=x$ and $f_{1}\left(\omega'\right)=f^{s}\left(\omega'\right)$
for all $\omega'\in\Omega\backslash\left\{ \omega\right\} $. Consider
$f_{2}\in\mathcal{\mathcal{F}_{S}\cap F}_{B}$ where $f_{2}\left(\omega\right)=y$
and $f_{2}\left(\omega'\right)=f^{s}\left(\omega'\right)$ for all
$\omega'\in\Omega\backslash\left\{ \omega\right\} $. Then $f_{1}\propto f_{2}$
and $f^{s}zq=af_{1}+\left(1-a\right)f_{2}$. \lemref{U-and-affine-on-prop}
gives $aU\left(f_{1}\right)+\left(1-a\right)U\left(f_{2}\right)=U\left(f^{s}zq\right)=U\left(f^{s}\right)$.
Due to $\tilde{U}_{I}\left(f^{s}\right)=\tilde{U}_{I}\left(f_{1}\right)=\tilde{U}_{I}\left(f_{2}\right)$
(\claimref{gamma1gamma2}), we have $aU_{C}\left(f_{1}\right)+\left(1-a\right)U_{C}\left(f_{2}\right)=U_{C}\left(f^{s}\right)$.
By \claimref{{uc for simple act}}, after simplification, we obtain
\begin{equation}
a\bar{u}\left(x,\omega\right)+\left(1-a\right)\bar{u}\left(y,\omega\right)=\bar{u}\left(z,\omega\right).\label{eq:sdeu1}
\end{equation}
Since $q\sim z$, \claimref{same_Uc} gives $U_{C}\left(q\right)=U_{C}\left(z\right)$.
\claimref{uc_seu_constant} then gives
\begin{equation}
au\left(x\right)+\left(1-a\right)u\left(y\right)=u\left(z\right).\label{eq:sdeu2}
\end{equation}
Therefore, \eqref{sdeu1} and \eqref{sdeu2} hold for any $x,y,z\in X\setminus B$
such that $\delta_{y}\succ\delta_{z}\succ\delta_{x}$ (with $\left(x,y,z\right)$-dependent
$a$). If $\bar{u}\left(x',\omega\right)=\bar{u}\left(y',\omega\right)$
for all $x',y'\in X\backslash B$, then let $a_{\omega}=0$ and find
$\eta_{\omega}\in\mathbb{R}$ such that $\bar{u}\left(z',\omega\right)=a_{\omega}u\left(z'\right)+\eta_{\omega}$
for all $z'\in X\setminus B$. If $\bar{u}\left(x',\omega\right)\neq\bar{u}\left(y',\omega\right)$
for some $x',y'\in X\backslash B$, then there exist $a_{\omega}>0$
and $\eta_{\omega}\in\mathbb{R}$ such that $\bar{u}\left(z',\omega\right)=a_{\omega}u\left(z'\right)+\eta_{\omega}$
for all $z'\in X\setminus B$.
\end{proof}
\begin{claim}
\label{claim:fullsupport}$a_{\omega}>0$ for all $\omega\in\Omega$.
\end{claim}
\begin{proof}
Fix any $\omega\in\Omega$. In the last part of proof of \claimref{alpha-eta},
$a_{\omega}=0$ only if $\bar{u}\left(x',\omega\right)=\bar{u}\left(y',\omega\right)$
for all $x',y'\in X\backslash B$. Take any distinct $x',y'\in X\backslash B$.
From the definition of $\bar{u}$ (\eqref{ubar}), $\bar{u}\left(x',\omega\right)=\bar{u}\left(y',\omega\right)$
only if $U\left(f_{x',\omega}\right)=U\left(f_{y',\omega}\right)$.
By \lemref{U-and-affine-on-prop}, this means $f_{x',\omega}\sim f_{y',\omega}$.
Suppose without loss of generality that $x'>y'$, then Monotonicity
in \axmref{standard} gives $x'\succ y'$, then \axmref{prize_independence}
gives $f_{x',\omega}\succ f_{y',\omega}$, a contradiction.
\end{proof}
\begin{claim}
\label{claim:parameter for SEU}If $\left\{ a_{\omega}\right\} _{\omega\in\Omega}$
and $\left\{ \eta_{\omega}\right\} _{\omega\in\Omega}$ satisfy $\bar{u}\left(z,\omega\right)=a_{\omega}u\left(z\right)+\eta_{\omega}$
for all $z\in X\backslash B$ and $\omega\in\Omega$, then $\sum_{\omega}a_{\omega}=1$
and $\sum_{\omega}\eta_{\omega}=0$.
\end{claim}
\begin{proof}
For any degenerate constant act $x\in\mathcal{F}_{X}$, consider the
split $x_{I}^{*}\in\mathcal{F}_{I}\cap\mathcal{F}_{B}$ such that
for each $\omega_{i}\in\Omega$, $x$ is replaced by $q_{i}\in\mathcal{F}_{P}\cap\mathcal{F}_{B}$
where $q_{i}\sim x$, and $\text{supp}\left(q_{j}\right)\cap\text{supp}\left(q_{k}\right)=\emptyset$
for any distinct $j,k$. The existence of $x_{I}^{*}$ is guaranteed
by the space $X$, Monotonicity (in \axmref{standard}), and Continuity
(in \axmref{standard}). Note that \claimref{same_Uc} gives $U_{C}\left(q_{i}\right)=U_{C}\left(\delta_{x}\right)$
for all $q_{i}$. By \claimref{{uc for simple act}},
\[
U_{C}\left(q_{i}\right)=\sum_{z\in\text{supp}\left(q_{i}\right)}q_{i}\left(z\right)u\left(z\right)=u\left(x\right)=U_{C}\left(\delta_{x}\right).
\]
Meanwhile, since $x_{I}^{*}$ is an informative act, using \claimref{{uc for informative act}},
we have
\begin{align*}
U_{C}\left(x_{I}^{*}\right) & =\sum_{\omega_{i}\in\Omega}\sum_{z\in\text{supp}\left(p_{I}\right)}x_{I}^{*}\left(z\mid\omega_{i}\right)\bar{u}\left(z,\omega_{i}\right)=\sum_{\omega_{i}\in\Omega}\sum_{z\in\text{supp}\left(q_{i}\right)}q_{i}\left(z\right)\bar{u}\left(z,\omega_{i}\right)\\
 & =\sum_{\omega_{i}\in\Omega}\sum_{z\in\text{supp}\left(q_{i}\right)}q_{i}\left(z\right)\left(a_{\omega_{i}}u\left(z\right)+\eta_{\omega_{i}}\right)=\sum_{\omega_{i}\in\Omega}\left(a_{\omega_{i}}u\left(x\right)+\eta_{\omega_{i}}\right)\\
 & =u\left(x\right)\cdot\sum_{\omega_{i}\in\Omega}a_{\omega_{i}}+\sum_{\omega_{i}\in\Omega}\eta_{\omega_{i}}.
\end{align*}
Moreover, from \claimref{{summation of vbar}} and \lemref{{menu representation}},
\begin{align*}
U\left(x_{I}^{*}\right)-U\left(\delta_{x}\right) & =\sum_{\omega\in\Omega}\left[\max_{\tilde{g}\in\tilde{A}}v\left(\tilde{g},\omega\right)\right]-\max_{\tilde{g}\in\tilde{A}}\left[\sum_{\omega\in\Omega}v\left(\tilde{g},\omega\right)\right]\\
U\left(x_{I}^{*}\right)-\sum_{\omega\in\Omega}\left[\max_{\tilde{g}\in\tilde{A}}v\left(\tilde{g},\omega\right)\right] & =U\left(\delta_{x}\right)-\max_{\tilde{g}\in\tilde{A}}\left[\sum_{\omega\in\Omega}v\left(\tilde{g},\omega\right)\right]\\
U\left(x_{I}^{*}\right)-\tilde{U}_{I}\left(x_{I}^{*}\right) & =U\left(\delta_{x}\right)-\tilde{U}_{I}\left(\delta_{x}\right)\\
U_{C}\left(x_{I}^{*}\right) & =U_{C}\left(\delta_{x}\right),
\end{align*}
where the third equality uses \claimref{gamma1gamma2} and the fourth
equality is by definition (\eqref{Uc-1}). Therefore,
\[
u\left(x\right)\cdot\sum_{\omega_{i}\in\Omega}a_{\omega_{i}}+\sum_{\omega_{i}\in\Omega}\eta_{\omega_{i}}=U_{C}\left(x_{I}^{*}\right)=U_{C}\left(\delta_{x}\right)=u\left(x\right),
\]
and this holds for any $x\in X$. By Monotonicity in \axmref{standard},
there exist $x,y\in X$ such that $u\left(x\right)\neq u\left(y\right)$.
So $\sum_{\omega}a_{\omega}=1$ and $\sum_{\omega}\eta_{\omega}=0$.
\end{proof}
Let 
\begin{equation}
\mu\left(\omega\right)=a_{\omega}\label{eq:mu}
\end{equation}
for all $\omega\in\Omega$. Then for any $x\in X\setminus B$, $u\left(x\right)=u\left(x\right)\cdot\sum_{\omega_{i}\in\Omega}a_{\omega_{i}}+\sum_{\omega_{i}\in\Omega}\eta_{\omega_{i}}=\sum_{\omega\in\Omega}\bar{u}\left(x,\omega\right)$.

\subsubsection*{Step 6: Demonstrating that $U_{C}$ is SEU}
\begin{defn}
\label{def:seu}Given any $\psi\in\Delta\left(\Omega\right)$ and
$\phi:X\rightarrow\mathbb{R}$, define the function $SEU_{\psi,\phi}:\mathcal{F}\rightarrow\mathbb{R}$
by 
\[
SEU_{\psi,\phi}\left(f\right)\coloneqq\sum_{\omega\in\Omega}\psi\left(\omega\right)\sum_{x\in\text{supp}\left(f\right)}f\left(x\mid\omega\right)\phi\left(x\right).
\]
\end{defn}
We want to show $U_{C}=SEU_{\mu,u}$ (recall that $U_{C}$ is defined
in \eqref{Uc-1}, $u$ is from \claimref{uc_seu_constant}, and $\mu$
is from \eqref{mu}).
\begin{claim}
\label{claim:informative_act_seu}For any informative act $I\in\mathcal{F}_{I}$,
$U_{C}\left(I\right)=SEU_{\mu,u}\left(I\right)$.
\end{claim}
\begin{proof}
For any informative act $I\in\mathcal{F}_{B}$, by \claimref{{uc for informative act}}
and \claimref{{parameter for SEU}},
\begin{align*}
U_{C}\left(I\right) & =\sum_{z\in\text{supp}\left(I\right)}\sum_{\omega\in\Omega}I\left(z\mid\omega\right)\bar{u}\left(z,\omega\right)=\sum_{z\in\text{supp}\left(I\right)}\sum_{\omega\in\Omega}I\left(z\mid\omega\right)\left(\mu\left(\omega\right)u\left(z\right)+\eta_{\omega}\right)\\
 & =\sum_{\omega\in\Omega}\mu\left(\omega\right)\sum_{z\in\text{supp}\left(I\right)}I\left(z\mid\omega\right)u\left(z\right).
\end{align*}
For any simple act $f^{s}\in\mathcal{F}_{S}\setminus\mathcal{F}_{B}$,
pick $q_{i}\in\Delta\left(X\setminus B\right)$ for each $i\in\left\{ 1,...,N\right\} $
such that $q_{i}\sim f^{s}\left(\omega_{i}\right)$, $\text{supp}\left(q_{i}\right)\cap\text{supp}\left(f^{s}\right)=\emptyset$,
and $\text{supp}\left(q_{i}\right)\cap\text{supp}\left(q_{j}\right)=\emptyset$
for distinct $i,j$. From $q_{i}\sim f^{s}\left(\omega_{i}\right)$,
\claimref{same_Uc} gives $U_{C}\left(q_{i}\right)=U_{C}\left(f^{s}\left(\omega_{i}\right)\right)$,
thus $\sum_{x\in\text{supp}\left(q_{i}\right)}q_{i}\left(x\right)u\left(x\right)=u\left(f^{s}\left(\omega_{i}\right)\right)$.
Iteratively use $q_{i}\in\Delta\left(X\setminus B\right)$ to replace
$f^{s}\left(\omega_{i}\right)$ for each $\omega_{i}\in\Omega$. Denote
the resulting act by $f^{*}$, note that $f^{*}\in\mathcal{F_{I}}\cap\mathcal{F}_{B}$.
By multiple uses of \axmref{prize_independence}, $f^{*}\sim f^{s}$,
so \claimref{same_Uc} gives $U_{C}\left(f^{*}\right)=U_{C}\left(f^{s}\right)$.
Then,
\begin{align*}
U_{C}\left(f^{*}\right) & =U_{C}\left(f^{s}\right)=\sum_{\omega\in\Omega}\mu\left(\omega\right)\sum_{z\in\text{supp}\left(I\right)}I\left(z\mid\omega\right)u\left(z\right)\\
 & =\sum_{\omega\in\Omega}\mu\left(\omega\right)\sum_{x\in\text{supp}\left(q_{i}\right)}q_{i}\left(x\right)u\left(x\right)=\sum_{\omega\in\Omega}\mu\left(\omega\right)u\left(f^{s}\left(\omega_{i}\right)\right).
\end{align*}
Since every informative act $I\in\mathcal{F}_{I}\setminus\mathcal{F}_{B}$
is equal to a convex combination of the simple acts in $S\left(I\right)$,
then by the same argument in \claimref{{uc for informative act}},
we have $U_{C}\left(I\right)=\sum_{\omega\in\Omega}\mu\left(\omega\right)\sum_{z\in\text{supp}\left(I\right)}I\left(z\mid\omega\right)u\left(z\right)$.
\end{proof}
\begin{lem}
\label{lem:SEU_Uc}For any act $f\in\mathcal{F}$, $U_{C}\left(f\right)=SEU_{\mu,u}\left(f\right)$. 
\end{lem}
\begin{proof}
\claimref{uc_seu_constant} has proved the case for constant acts
and \claimref{informative_act_seu} has proved the case for informative
acts. Fix any other $f\in\mathcal{F}$, generate a $f_{I}^{*}\in I^{*}\left(f\right)$
and denote the lottery replacement of $z\in\text{supp}\left(f\left(\omega_{i}\right)\right)$
under $\omega_{i}\in\Omega$ by $q_{i}^{z}\in\Delta\left(X\right)$.
For all $i\in\left\{ 1,...,N\right\} $ and $z\in\text{supp}\left(f\left(\omega_{i}\right)\right)$,
recall from the definition of $f_{I}^{*}$ that $q_{i}^{z}\sim z$,
so \claimref{same_Uc} gives $U_{C}\left(q_{i}^{z}\right)=U_{C}\left(\delta_{z}\right)$,
and \claimref{uc_seu_constant} further gives $\sum_{x\in\text{supp}\left(q_{i}\right)}q_{i}^{z}\left(x\right)u\left(x\right)=u\left(z\right)$.
Therefore,
\begin{align*}
U_{C}\left(f\right) & =U_{C}\left(f_{I}^{*}\right)=\sum_{\omega\in\Omega}\mu\left(\omega\right)\sum_{z\in\text{supp}\left(f_{I}^{*}\right)}f_{I}^{*}\left(z\mid\omega\right)u\left(z\right)\\
 & =\sum_{\omega\in\Omega}\mu\left(\omega\right)\sum_{z\in\text{supp}\left(f\right)}f\left(z\mid\omega\right)\left[\sum_{x\in\text{supp}\left(q_{i}\right)}q_{i}^{z}\left(x\right)u\left(x\right)\right]\\
 & =\sum_{\omega\in\Omega}\mu\left(\omega\right)\sum_{z\in\text{supp}\left(f\right)}f\left(z\mid\omega\right)u\left(z\right)=SEU_{\mu,u}\left(f\right).
\end{align*}
where the first equality is obtained from rearranging \claimref{{summation of vbar}}
for $U\left(f\right)-\tilde{U}_{I}\left(f\right)=U\left(f_{I}^{*}\right)-\tilde{U}_{I}\left(f_{I}^{*}\right)$
and using the definition of $U_{C}$ (\eqref{Uc-1}); the second equality
is due to $f_{I}^{*}\in\mathcal{F}_{I}$ and \claimref{informative_act_seu};
the third equality is by the definition of $f_{I}^{*}$; and the fourth
equality is established earlier in the proof.
\end{proof}

\subsubsection*{Step 7: Deriving $V$}

Recall that $\mu$ has full support by \claimref{fullsupport}, so
dividing by $\mu\left(\omega\right)$ is well-defined. From \lemref{{menu representation}}
we get
\begin{align}
\tilde{U}_{I}\left(f\right) & =\sum_{x\in\text{supp}\left(f\right)}\sup_{\tilde{g}\in\tilde{A}}\sum_{\omega\in\Omega}f\left(x\mid\omega\right)v\left(\tilde{g},\omega\right)\nonumber \\
 & =\sum_{x\in\text{supp}\left(f\right)}\sup_{\tilde{g}\in\tilde{A}}\sum_{\omega\in\Omega}\left[\frac{\sum_{\omega'\in\Omega}\mu\left(\omega'\right)f\left(x\mid\omega'\right)}{\sum_{\omega'\in\Omega}\mu\left(\omega'\right)f\left(x\mid\omega'\right)}\frac{\mu\left(\omega\right)}{\mu\left(\omega\right)}\right]f\left(x\mid\omega\right)v\left(\tilde{g},\omega\right)\nonumber \\
 & =\sum_{x\in\text{supp}\left(f\right)}\sum_{\omega'\in\Omega}\mu\left(\omega'\right)f\left(x\mid\omega'\right)\sup_{\tilde{g}\in\tilde{A}}\left[\sum_{\omega\in\Omega}\frac{\mu\left(\omega\right)f\left(x\mid\omega\right)}{\sum_{\omega'\in\Omega}\mu\left(\omega'\right)f\left(x\mid\omega'\right)}\frac{v\left(\tilde{g},\omega\right)}{\mu\left(\omega\right)}\right]\nonumber \\
 & =\sum_{x\in\text{supp}\left(f\right)}\sum_{\omega\in\Omega}\mu\left(\omega\right)f\left(x\mid\omega\right)\sup_{\tilde{g}\in\tilde{A}}\left[\sum_{\omega\in\Omega}\mu_{x,f}\left(\omega\right)\frac{v\left(\tilde{g},\omega\right)}{\mu\left(\omega\right)}\right].\label{eq:forICrep}
\end{align}
Let
\begin{equation}
V\left(\mu_{x,f}\right)\coloneqq\sup_{\tilde{g}\in\tilde{A}}\left[\sum_{\omega\in\Omega}\mu_{x,f}\left(\omega\right)\frac{v\left(\tilde{g},\omega\right)}{\mu\left(\omega\right)}\right].\label{eq:V}
\end{equation}
Finally, because $U\left(f\right)=U_{C}\left(f\right)+\tilde{U}_{I}\left(f\right)$
by definition (\eqref{Uc-1}), we have 
\[
U\left(f\right)=U_{C}\left(f\right)+\sum_{x\in\text{supp}\left(f\right)}\sum_{\omega\in\Omega}\mu\left(\omega\right)f\left(x\mid\omega\right)V\left(\mu_{x,f}\right)
\]
 as desired.

\subsection{Proof of Theorem \ref{thm:IC-2}}

We prove ``only if'' as the other direction is straightforward after
utilizing the fact that an IC representation attains its $\max$,
which precludes the case of ``vertical hyperplane.''

Since $\succsim$ admits an IC-V representation, the proof of \thmref{IC-1}
can be replayed until \eqref{forICrep}. Note that if
\[
\max_{\tilde{g}\in\tilde{A}}\left[\sum_{\omega\in\Omega}\mu_{x,f}\left(\omega\right)\frac{v\left(\tilde{g},\omega\right)}{\mu\left(\omega\right)}\right]
\]
is undefined (i.e., there is no argmax $\tilde{g}$) for some $\mu_{x,f}\in\Delta\left(\Omega\right)$,
then $V\left(\mu_{x,f}\right)$ (\eqref{V}) will have infinite (or
negative infinite) slopes at at least one of the degenerate posteriors,
which contradicts the fact that $V$ is Lipschitz continuous. This
corresponds to the known issue of ``vertical'' supporting hyperplanes
of $V\left(\cdot\right)$ towards the boundary. Hence
\[
\sup_{\tilde{g}\in\tilde{A}}\left[\sum_{\omega\in\Omega}\mu_{x,f}\left(\omega\right)\frac{v\left(\tilde{g},\omega\right)}{\mu\left(\omega\right)}\right]=\max_{\tilde{g}\in\tilde{A}}\left[\sum_{\omega\in\Omega}\mu_{x,f}\left(\omega\right)\frac{v\left(\tilde{g},\omega\right)}{\mu\left(\omega\right)}\right].
\]

Excluding never-optimal $\tilde{g}$ from $\tilde{A}$ (that is, $\tilde{g}$
that is never the argmax for any $\mu_{x,f}\in\Delta\left(\Omega\right)$),
the previous argument implies that $\max_{\tilde{g}\in\tilde{A}}v\left(\tilde{g},\omega\right)$
and $\min_{\tilde{g}\in\tilde{A}}v\left(\tilde{g},\omega\right)$
are well-defined (attained) for all $\omega\in\Omega$ (these are
the maximum and minimum intercepts of the hyperplanes $v\left(\tilde{g},\cdot\right)$).
So $u^{+}\coloneqq\max_{\tilde{g}\in\tilde{A},\omega\in\Omega}\frac{v\left(\tilde{g},\omega\right)}{\mu\left(\omega\right)}$
and $u^{-}\coloneqq\min_{\tilde{g}\in\tilde{A},\omega\in\Omega}\frac{v\left(\tilde{g},\omega\right)}{\mu\left(\omega\right)}$
are well-defined

Suppose for now that $u^{-}<u^{+}$. Let $\beta$ satisfy $\beta\left(u\left(1\right)-u\left(0\right)\right)=u^{+}-u^{-}$,
where $u$ is the Bernoulli utility function already obtained. Observe
that $\beta>0$ and $\frac{1}{\beta}\frac{v\left(\tilde{g},\omega\right)}{\mu\left(\omega\right)}\in\left[u\left(0\right),u\left(1\right)\right]$
for any $\tilde{g}\in\tilde{A}$ and $\omega\in\Omega$. For any $\tilde{g}\in\tilde{A}$
and $\omega\in\Omega$, let (define) $p_{\omega}\left(\tilde{g}\right)$
satisfy 
\[
p_{\omega}\left(\tilde{g}\right)u\left(1\right)+\left(1-p_{\omega}\left(g\right)\right)u\left(0\right)=\frac{1}{\beta}\frac{v\left(\tilde{g},\omega\right)}{\mu\left(\omega\right)}.
\]
Construct $A\subseteq\mathcal{F}$ by including $g\in\mathcal{F}$
in $A$ if and only if there exists $\tilde{g}\in\tilde{A}$ such
that $g\left(\omega\right)=p_{\omega}\left(\tilde{g}\right)\delta_{1}+\left(1-p_{\omega}\left(\tilde{g}\right)\right)\delta_{0}$
for all $\omega\in\Omega$. Continuing from \eqref{forICrep}, we
now have
\begin{align*}
\tilde{U}_{I}\left(f\right) & =\sum_{x\in\text{supp}\left(f\right)}\sum_{\omega\in\Omega}\mu\left(\omega\right)f\left(x\mid\omega\right)\max_{\tilde{g}\in\tilde{A}}\left[\sum_{\omega\in\Omega}\mu_{x,f}\left(\omega\right)\frac{v\left(\tilde{g},\omega\right)}{\mu\left(\omega\right)}\right]\\
 & =\sum_{x\in\text{supp}\left(f\right)}\sum_{\omega\in\Omega}\mu\left(\omega\right)f\left(x\mid\omega\right)\max_{g\in A}\left[\sum_{\omega\in\Omega}\mu_{x,f}\left(\omega\right)\beta\left[g\left(\omega\right)\left(1\right)u_{1}\left(1\right)+g\left(\omega\right)\left(0\right)u_{1}\left(0\right)\right]\right]\\
 & =\beta\sum_{x\in\text{supp}\left(f\right)}\sum_{\omega\in\Omega}\mu\left(\omega\right)f\left(x\mid\omega\right)\left[\max_{g\in A}\sum_{\omega\in\Omega}\mu_{x,f}\left(\omega\right)\sum_{y\in\text{supp}\left(g\right)}g\left(y\mid\omega\right)u\left(y\right)\right]
\end{align*}
Because $U\left(f\right)=U_{C}\left(f\right)+\tilde{U}_{I}\left(f\right)$
by definition, we now have
\[
U\left(f\right)=U_{C}\left(f\right)+\beta U_{I}\left(f\right)
\]
where $U_{I}\left(f\right)$ is from \eqref{UI}. Let $\alpha>0$
satisfy $\beta=\frac{1-\alpha}{\alpha}$ and multiply $U\left(f\right)$
by $\alpha$, we now obtain an IC representation with parameters $\left(\mu,u,\alpha,A\right)$.

Suppose instead that $u^{-}=u^{+}$ (this is the special case where
information has no value). Then for any $\omega\in\Omega$, we have
$\frac{v\left(\tilde{g},\omega\right)}{\mu\left(\omega\right)}=c\in\mathbb{R}$
for all $\tilde{g}\in\tilde{A}$. So every act produces the same information
value and therefore $f\succsim g$ if and only $U_{C}\left(f\right)\geq U_{C}\left(g\right)$.
With freedom, we conclude that $\succsim$ has a IC representation
with parameters $\left(\mu,u,\alpha=1,A=\left\{ \delta_{0}\right\} \right)$
($A$ still needs to be well-specified for the $\max$ operation to
be well-defined).

\subsection{Proof of Proposition \ref{prop:identificationV}}

The ``if'' direction is straightforward noting that for any $f$,
\begin{align*}
 & \sum_{x\in\text{supp}\left(f\right)}\sum_{\omega\in\Omega}\mu\left(\omega\right)f\left(x\mid\omega\right)\left[\lambda\cdot\mu_{x,f}\right]\\
 & =\sum_{x\in\text{supp}\left(f\right)}\sum_{\omega\in\Omega}\mu\left(\omega\right)f\left(x\mid\omega\right)\left[\sum_{\omega'\in\Omega}\lambda_{\omega'}\frac{\mu\left(\omega'\right)f\left(x\mid\omega'\right)}{\sum_{\omega''\in\Omega}\mu\left(\omega''\right)f\left(x\mid\omega''\right)}\right]\\
 & =\sum_{x\in\text{supp}\left(f\right)}\sum_{\omega'\in\Omega}\left[\lambda_{\omega'}\mu\left(\omega'\right)f\left(x\mid\omega'\right)\right]\\
 & =\sum_{\omega'\in\Omega}\sum_{x\in\text{supp}\left(f\right)}\left[\lambda_{\omega'}\mu\left(\omega'\right)f\left(x\mid\omega'\right)\right]\\
 & =\sum_{\omega'\in\Omega}\left[\lambda_{\omega'}\mu\left(\omega'\right)\right]\\
 & =\lambda\cdot\mu,
\end{align*}
which is independent of $f$.

We show the ``only if'' direction. Suppose $\succsim$ admits an
IC-V representation $\left(\mu_{1},u_{1},V_{1}\right)$ and it also
admits an IC-V representation $\left(\mu_{2},u_{2},V_{2}\right)$.
For any $i\in\left\{ 1,2\right\} $ and $f\in\mathcal{F}$, let $U_{I}^{i}\left(f\right)=\sum_{x\in\text{supp}\left(f\right)}\sum_{\omega\in\Omega}\mu_{i}\left(\omega\right)f\left(x\mid\omega\right)V_{i}\left(\mu_{x,f}\right)$.

For any $x\in X$, 
\[
U_{i}\left(\delta_{x}\right)=SEU_{\mu_{i}u_{i}}\left(\delta_{x}\right)+V_{i}\left(\mu_{i}\right)=SEU_{\mu_{i},u_{i}+V_{i}\left(\mu_{i}\right)}\left(\delta_{x}\right)
\]
for $i\in\left\{ 1,2\right\} $. Since $\delta_{x}\succsim\delta_{y}$
if and only if $U_{i}\left(\delta_{x}\right)\geq U_{i}\left(\delta_{y}\right)$
for all $i\in\left\{ 1,2\right\} $, then by the classical SEU identification
result, there exists $\kappa>0$ and $c'\in\mathbb{R}$ such that
$u_{2}\left(x\right)+V_{2}\left(\mu_{2}\right)=\kappa u_{1}\left(x\right)+V_{1}\left(\mu_{1}\right)+c'$
for any $x\in X$. Let $c=V_{1}\left(\mu_{1}\right)+c'-V_{2}\left(\mu_{2}\right)$.
We have $u_{2}\left(x\right)=\kappa u_{1}\left(x\right)+c$. Therefore,
for prize $1$ and prize $0$, we have $u_{2}\left(1\right)-u_{2}\left(0\right)=\kappa\left(u_{1}\left(1\right)-u_{1}\left(0\right)\right)$. 

We now argue that $\mu_{1}=\mu_{2}$. For any $x\in X$, there exists
a unique $a_{x}\in\mathbb{R}$ such that 
\begin{equation}
u_{i}\left(x\right)=a_{x}u_{i}\left(1\right)+\left(1-a_{x}\right)u_{i}\left(0\right)=a_{x}\left(u_{i}\left(1\right)-u_{i}\left(0\right)\right)+u_{i}\left(0\right)\label{eq:ax}
\end{equation}
 for any $i\in\left\{ 1,2\right\} $. Let $N=|\Omega|$. For the following,
we abuse notation and write $z$ and $\delta_{z}$ interchangeably. 
\begin{enumerate}
\item Let $f_{1}\in\mathcal{F}$ be such that $f_{1}\left(\omega_{1}\right)=5$,
$f_{1}\left(\omega_{2}\right)=6$ and $f_{1}\left(\omega_{i}\right)=-1$
for any $\omega_{i}\in\Omega\setminus\left\{ \omega_{1},\omega_{2}\right\} $. 
\item Let $f_{2}\in\mathcal{F}$ be such that $f_{2}\left(\omega_{1}\right)=6$,
$f_{2}\left(\omega_{2}\right)=5$ and $f_{2}\left(\omega_{i}\right)=-1$
for any $\omega_{i}\in\Omega\setminus\left\{ \omega_{1},\omega_{2}\right\} $. 
\item Let $g\in\mathcal{F}$ be such that $g\left(\omega_{1}\right)=3$,
$g\left(\omega_{2}\right)=4$ and $g\left(\omega_{i}\right)=-1$ for
any $\omega_{i}\in\Omega\setminus\left\{ \omega_{1},\omega_{2}\right\} $. 
\item Let $h\in\mathcal{F}$ be such that $h\left(\omega_{1}\right)=7$,
$h\left(\omega_{2}\right)=8$ and $h\left(\omega_{i}\right)=-1$ for
any $\omega_{i}\in\Omega\setminus\left\{ \omega_{1},\omega_{2}\right\} $.
\end{enumerate}
There exists $t_{1},t_{2}\in\left(0,1\right)$ such that $f_{1}\sim t_{1}g+\left(1-t_{1}\right)h$
and $f_{2}\sim t_{2}g$+$\left(1-t_{2}\right)h$. Then, for any $i\in\left\{ 1,2\right\} $:
$U_{i}\left(f_{1}\right)=U_{i}\left(t_{1}g+\left(1-t_{1}\right)h\right)$
and $U_{i}\left(f_{1}\right)=t_{1}U_{i}\left(g\right)+\left(1-t_{1}\right)U_{i}\left(h\right)$.
Because $f_{1},g,h$ have equal information value i.e., $U_{I}^{i}\left(f_{1}\right)=U_{I}^{i}\left(g\right)=U_{I}^{i}\left(h\right)$,
$\mu_{i}\left(\omega_{1}\right)u_{i}\left(5\right)+\mu_{i}\left(\omega_{2}\right)u_{i}\left(6\right)=t_{1}\left(\mu_{i}\left(\omega_{1}\right)u_{i}\left(3\right)+\mu_{i}\left(\omega_{2}\right)u\left(4\right)\right)+\left(1-t_{1}\right)\left(\mu_{i}\left(\omega_{1}\right)u_{i}\left(7\right)+\mu_{i}\left(\omega_{2}\right)u_{i}\left(8\right)\right)$,
so \eqref{ax} gives 
\[
\mu_{i}\left(\omega_{1}\right)a_{5}+\mu_{i}\left(\omega_{2}\right)a_{6}=t_{1}\left(\mu_{i}\left(\omega_{1}\right)a_{3}+\mu_{i}\left(\omega_{2}\right)a_{4}\right)+\left(1-t_{1}\right)\left(\mu_{i}\left(\omega_{1}\right)a_{7}+\mu_{i}\left(\omega_{2}\right)a_{8}\right).
\]
Similarly, because $f_{2},g,h$ have equal information value, we obtain
\[
\mu_{i}\left(\omega_{1}\right)a_{6}+\mu_{i}\left(\omega_{2}\right)a_{5}=t_{2}\left(\mu_{i}\left(\omega_{1}\right)a_{3}+\mu_{i}\left(\omega_{2}\right)a_{4}\right)+\left(1-t_{2}\right)\left(\mu_{i}\left(\omega_{1}\right)a_{7}+\mu_{i}\left(\omega_{2}\right)a_{8}\right).
\]

Notice we obtain the same two equations for two unknowns $\mu_{i}\left(\omega_{1}\right),\mu_{i}\left(\omega_{2}\right)$
for each of $i\in\left\{ 1,2\right\} $, therefore we have $\mu_{1}\left(\omega_{1}\right)=\mu_{2}\left(\omega_{1}\right)$
and $\mu_{1}\left(\omega_{2}\right)=\mu_{2}\left(\omega_{2}\right)$.
Repeating the same arguments for other $\omega$'s gives $\mu_{1}\left(\omega\right)=\mu_{2}\left(\omega\right)$
for all $\omega\in\Omega$, hence $\mu_{1}=\mu_{2}$. We denote the
unique prior as $\mu$ from now on.

Next, let $I^{*}$ be a simple act (see \subsecref{notations-objects})
such that $x\succ0$ for any $x\in\text{supp}\left(I^{*}\right)$.
We argue that $U_{I}^{2}\left(I^{*}\right)-U_{I}^{2}\left(\delta_{0}\right)=\kappa\left(U_{I}^{1}\left(I^{*}\right)-U_{I}^{1}\left(\delta_{0}\right)\right)$.
There exists $\epsilon\in\left(0,1\right)$ and constant act (lottery)
$p\in\mathcal{F}_{P}$ such that $\epsilon f+\left(1-\epsilon\right)\delta_{0}\sim p$.
Then for any $i\in\left\{ 1,2\right\} $,
\[
U^{i}\left(\epsilon f+\left(1-\epsilon\right)\delta_{0}\right)=U^{i}\left(p\right)
\]
\[
\epsilon U_{c}^{i}\left(I^{*}\right)+\left(1-\epsilon\right)U_{c}^{i}\left(\delta_{0}\right)+\epsilon U_{I}^{i}\left(I^{*}\right)+\left(1-\epsilon\right)U_{I}^{i}\left(\delta_{0}\right)=U_{c}^{i}\left(p\right)+U_{I}^{i}\left(p\right)
\]
\[
\epsilon\left(U_{I}^{i}\left(I^{*}\right)-U_{I}^{i}\left(\delta_{0}\right)\right)=U_{c}^{i}\left(p\right)-\epsilon U_{c}^{i}\left(I^{*}\right)-\left(1-\epsilon\right)U_{c}^{i}\left(\delta_{0}\right).
\]
The third equality follows from $U_{I}^{i}\left(\delta_{0}\right)=U_{I}^{i}\left(p\right)$.
Therefore, 
\[
\frac{\epsilon\left(U_{I}^{1}\left(I^{*}\right)-U_{I}^{1}\left(\delta_{0}\right)\right)}{\epsilon\left(U_{I}^{2}\left(I^{*}\right)-U_{I}^{2}\left(\delta_{0}\right)\right)}=\frac{U_{c}^{1}\left(p\right)-\epsilon U_{c}^{1}\left(I^{*}\right)-\left(1-\epsilon\right)U_{c}^{1}\left(\delta_{0}\right)}{U_{c}^{2}\left(p\right)-\epsilon U_{c}^{2}\left(I^{*}\right)-\left(1-\epsilon\right)U_{c}^{2}\left(\delta_{0}\right)}=\frac{u_{2}\left(1\right)-u_{2}\left(0\right)}{u_{1}\left(1\right)-u_{1}\left(0\right)}=\kappa
\]

Let $\mathcal{F}^{*}=\left\{ f\in\mathcal{F}:I^{*}\succ f\succ\delta_{0},f\cap I^{*}=\emptyset,f\cap\delta_{0}=\emptyset\right\} $.
Then, for any $f\in\mathcal{F}^{*}$, we show $U_{I}^{2}\left(f\right)-\kappa U_{I}^{1}\left(f\right)=U_{I}^{2}\left(\delta_{0}\right)-\kappa U_{I}^{1}\left(\delta_{0}\right)$.
Since there exists a unique $a_{f}\in\left(0,1\right)$ such that
$a_{f}I^{*}+\left(1-a_{f}\right)\delta_{0}\sim f$. Then, for any
$i\in\left\{ 1,2\right\} $:
\[
U^{i}\left(f\right)=U\left(a_{f}I^{*}+\left(1-a_{f}\right)\delta_{0}\right)
\]
\[
U_{c}^{i}\left(f\right)+U_{I}^{i}\left(f\right)=a_{f}U_{c}^{i}\left(I^{*}\right)+a_{f}U_{I}^{i}\left(I^{*}\right)+\left(1-a_{f}\right)U_{c}^{i}\left(\delta_{0}\right)+\left(1-a_{f}\right)U_{I}^{i}\left(\delta_{0}\right)
\]
\[
U_{I}^{i}\left(f\right)-U_{I}^{i}\left(\delta_{0}\right)=a_{f}\left(U_{I}^{i}\left(I^{*}\right)-U_{I}^{i}\left(\delta_{0}\right)\right)+a_{f}U_{c}^{i}\left(I^{*}\right)+\left(1-a_{f}\right)U_{c}^{i}\left(\delta_{0}\right)-U_{c}^{i}\left(f\right).
\]

Since there exists $k\in\mathbb{R}$ such that $a_{f}U_{c}^{i}\left(I^{*}\right)+\left(1-a_{f}\right)U_{c}^{i}\left(\delta_{0}\right)-U_{c}^{i}\left(f\right)=k\left(u_{i}\left(1\right)-u_{i}\left(0\right)\right)$
for any $i\in\left\{ 1,2\right\} $, we have $\kappa\left(a_{f}U_{c}^{1}\left(I^{*}\right)+\left(1-a_{f}\right)U_{c}^{1}\left(\delta_{0}\right)-U_{c}^{1}\left(f\right)\right)=a_{f}U_{c}^{1}\left(I^{*}\right)+\left(1-a_{f}\right)U_{c}^{1}\left(\delta_{0}\right)-U_{c}^{1}\left(f\right)$.
Then, $\kappa U_{I}^{1}\left(f\right)-\kappa U_{I}^{1}\left(\delta_{0}\right)-\left(U_{I}^{2}\left(f\right)-\kappa U_{I}^{2}\left(\delta_{0}\right)\right)=\kappa a_{f}\left(U_{I}^{1}\left(I^{*}\right)-U_{I}^{1}\left(\delta_{0}\right)\right)-a_{f}\left(U_{I}^{2}\left(I^{*}\right)-U_{I}^{2}\left(\delta_{0}\right)\right)$,
which leads to 
\[
U_{I}^{2}\left(f\right)-\kappa U_{I}^{1}\left(f\right)=U_{I}^{2}\left(\delta_{0}\right)-\kappa U_{I}^{1}\left(\delta_{0}\right)
\]
due to $U_{I}^{2}\left(I^{*}\right)-U_{I}^{2}\left(\delta_{0}\right)=\kappa\left(U_{I}^{1}\left(I^{*}\right)-U_{I}^{1}\left(\delta_{0}\right)\right)$.

Next, we show there exist $\lambda\in\mathbb{R}^{|\Omega|}$ and $b\in\mathbb{R}$
such that $V_{2}\left(\pi\right)=\kappa V_{1}\left(\pi\right)+\lambda\cdot\pi+b$.
It suffices to show $V_{2}\left(\pi\right)-\kappa V_{1}\left(\pi\right)$
is an affine function of $\pi$. Fix any two beliefs $\pi_{1},\pi_{2}\in\Delta\left(\Omega\right)$
and $\gamma\in\left[0,1\right]$, let $\pi=\gamma\pi_{1}+\left(1-\gamma\right)\pi_{2}$.
Then, there exists a $\pi'\in\Delta\left(\Omega\right)$ and $\epsilon\in\left(0,1\right)$
such that $\epsilon\pi+\left(1-\epsilon\right)\pi'=\mu_{0}$ and $\pi'=\frac{\mu_{0}-\epsilon\pi}{1-\epsilon}$.
Take $f\in\mathcal{F^{*}}$ such that the induced posterior distribution
for $f$ is $\pi$ with probability $\epsilon$ and $\pi'$ with probability
$1-\epsilon$. Take $g\in\mathcal{F^{*}}$ such that the induced posterior
distribution for $g$ is $\pi_{1}$ with probability $\gamma\epsilon$,
$\pi_{2}$ with probability $\left(1-\gamma\right)\epsilon$ and $\pi'$
with probability $1-\epsilon$. The process is feasible since there
are infinite prizes between 0 and $\min_{\succsim}\text{supp\ensuremath{\left(I^{*}\right)}}$.
Therefore, using $U_{I}^{2}\left(h\right)-\kappa U_{I}^{1}\left(h\right)=U_{I}^{2}\left(\delta_{0}\right)-\kappa U_{I}^{1}\left(\delta_{0}\right)$
for any $f\in\mathcal{F}^{*}$, we get $U_{I}^{2}\left(f\right)-\kappa U_{I}^{1}\left(f\right)=U_{I}^{2}\left(g\right)-\kappa U_{I}^{1}\left(g\right)$,
which implies $\epsilon V_{2}\left(\pi\right)-\kappa\epsilon V_{1}\left(\pi\right)=\gamma\epsilon V_{2}\left(\pi_{1}\right)+\left(1-\gamma\right)\epsilon V_{2}\left(\pi_{2}\right)-\kappa\left(\text{\ensuremath{\epsilon V_{1}\left(\pi_{1}\right)}+\ensuremath{\left(1-\gamma\right)\epsilon V_{1}\left(\pi_{2}\right)}}\right)$
(due to $U_{I}^{i}\left(f\right)=\epsilon V_{2}\left(\pi\right)+\left(1-\epsilon\right)V_{2}\left(\pi'\right)$
and $U_{I}^{i}\left(g\right)=\gamma\epsilon V_{2}\left(\pi_{1}\right)+\left(1-\gamma\right)\epsilon V_{2}\left(\pi'\right)+\left(1-\epsilon\right)V_{2}\left(\pi'\right)$
for $i\in\left\{ 1,2\right\} $) and then $V_{2}\left(\gamma\pi_{1}+\left(1-\gamma\right)\pi_{2}\right)-\kappa\epsilon V_{1}\left(\gamma\pi_{1}+\left(1-\gamma\right)\pi_{2}\right)=\gamma\left(V_{2}\left(\pi_{1}\right)-\kappa V_{1}\left(\pi\right)\right)+\left(1-\gamma\right)\left(V_{2}\left(\pi_{2}\right)-\kappa V_{1}\left(\pi_{2}\right)\right).$By
the arbitrariness of $\pi_{1},\pi_{2}$ and $\gamma$, $V_{2}\left(\cdot\right)-\kappa V_{1}\left(\cdot\right)$
is an affine function over $\Delta\left(\Omega\right)$. So there
exist $\lambda\in\mathbb{R}^{|\Omega|}$ and $b\in\mathbb{R}$ such
that $V_{2}\left(\pi\right)=\kappa V_{1}\left(\pi\right)+\lambda\cdot\pi+b$.

\subsection{Proof of Proposition \ref{prop:identificationA}}

\paragraph*{Only if:}

Suppose $\succsim$ admits IC representations $\left(\mu_{1},u_{1},\alpha_{1},A_{1}\right)$
and $\left(\mu_{2},u_{2},\alpha_{2},A_{2}\right)$. Then by \thmref{IC-2},
it also admits IC-V representations $\left(\mu_{1},u_{1},V_{1}^{*}\right)$
and $\left(\mu_{2},u_{2},V_{2}^{*}\right)$, where $V_{i}^{*}$ satisfy
\eqref{V*}. Therefore, by \propref{identificationV}, there exist
$\kappa\in\mathbb{R}_{++}$, $\lambda\in\mathbb{R}^{\left|\Omega\right|}$,
and $a,b\in\mathbb{R}$ such that $u_{2}=\kappa u_{1}+a$ and $V_{2}^{*}\left(\pi\right)=\kappa V_{1}^{*}\left(\pi\right)+\lambda\cdot\pi+b$
for all $\pi\in\Delta\left(\Omega\right)$.

\paragraph*{If:}

Suppose $\succsim$ admits an IC representation $\left(\mu_{1},u_{1},\alpha_{1},A_{1}\right)$.
By \thmref{IC-2}, it admits an IC-V representation $\left(\mu_{1},u_{1},V_{1}^{*}\right)$
where $V_{1}^{*}$ is Lipschitz continuous. Consider $\mu_{2}=\mu_{1}$,
$u_{2}=\kappa u_{1}+a$ and $V_{2}^{*}\left(\pi\right)=\kappa V_{1}^{*}\left(\pi\right)+\lambda\cdot\pi+b$
for all $\pi\in\Delta\left(\Omega\right)$ where $\kappa\in\mathbb{R}_{++}$,
$\lambda\in\mathbb{R}^{\left|\Omega\right|}$, and $a,b\in\mathbb{R}$.
Note that the affine operation guarantees that $V_{2}^{*}\left(\pi\right)$
is also Lipschitz continuous. By \propref{identificationV}, $\succsim$
admits an IC-V representation $\left(\mu_{2},u_{2},V_{2}^{*}\right)$.
By \thmref{IC-2}, $\succsim$ admits an IC representation $\left(\mu_{2},u_{2},\alpha_{2},A_{2}\right)$.

\subsection{Proof of Proposition \ref{prop:icad}}

Given the qualifying IC representation $\left(\mu,u,\alpha,A\right)$,
we first note that $A$ is finite so the sup in $U_{I}\left(f\right)$
is attained (i.e., it is a max). Dividing $\alpha U_{C}\left(f\right)+\left(1-\alpha\right)U_{I}\left(f\right)$
by $\alpha$ gives $U\left(f\right)=U_{C}\left(f\right)+\frac{1-\alpha}{\alpha}U_{I}\left(f\right)$.
Then, the continuation value $U_{I}\left(f\right)$ can be simplified:
\begin{align*}
 & U_{I}\left(f\right)\\
 & =\sum_{x\in\text{supp}\left(f\right)}\sum_{\omega\in\Omega}\mu\left(\omega\right)f\left(x\mid\omega\right)\max_{a_{\iota,\omega}\in A}\left[\mu_{x,f}\left(\omega\right)\underbrace{u\left(u^{-1}\left(\frac{\alpha}{1-\alpha}u\left(\iota\right)\right)\right)}_{\frac{\alpha}{1-\alpha}u\left(\iota\right)}+\sum_{\omega'\ne\omega}\mu_{x,f}\left(\omega'\right)\underbrace{u\left(0\right)}_{\text{normalized to }0}\right]\\
 & =\sum_{x\in\text{supp}\left(f\right)}\sum_{\omega\in\Omega}\mu\left(\omega\right)f\left(x\mid\omega\right)\max_{\omega\in\Omega}\mu_{x,f}\left(\omega\right)\frac{\alpha}{1-\alpha}u\left(\iota\right)\\
 & =\sum_{x\in\text{supp}\left(f\right)}\sum_{\omega\in\Omega}\mu\left(\omega\right)f\left(x\mid\omega\right)\max_{\omega\in\Omega}\underbrace{\frac{\mu\left(\omega\right)f\left(x\mid\omega\right)}{\sum_{\omega'\in\Omega}\mu\left(\omega'\right)f\left(x\mid\omega'\right)}}_{\text{Bayes rule}}\frac{\alpha}{1-\alpha}u\left(\iota\right)\\
 & \sum_{x\in\text{supp}\left(f\right)}\underbrace{\frac{\sum_{\omega\in\Omega}\mu\left(\omega\right)f\left(x\mid\omega\right)}{\sum_{\omega'\in\Omega}\mu\left(\omega'\right)f\left(x\mid\omega'\right)}}_{\text{1}}\max_{\omega\in\Omega}\mu\left(\omega\right)f\left(x\mid\omega\right)\frac{\alpha}{1-\alpha}u\left(\iota\right)\\
 & =\frac{\alpha}{1-\alpha}u\left(\iota\right)\underbrace{\sum_{x\in X}\max_{\omega\in\Omega}\mu\left(\omega\right)f\left(x\mid\omega\right)}_{\mathbb{I}\left(f,\mu\right)\text{ by definition}}.
\end{align*}
So 
\[
U\left(f\right)=U_{C}\left(f\right)+\frac{1-\alpha}{\alpha}U_{I}\left(f\right)=U_{C}\left(f\right)+u\left(\iota\right)\mathbb{I}\left(f,\mu\right).
\]

\subsection{Proof of Proposition \ref{prop:same_rho}}

\paragraph*{If: }

Take $\mu$ as given. Let $U_{C}\left(f;\rho\right)\coloneqq\sum_{\omega\in\Omega}\mu\left(\omega\right)\sum_{x\in X}f\left(x\mid\omega\right)u_{\rho}\left(x\right)$
where $u_{\rho}\left(x\right)=1-e^{-\rho x}$ and let $U\left(f;\rho,\iota\right)=U_{C}\left(f;\rho\right)+u_{\rho}\left(\iota\right)\mathbb{I}\left(f,\mu\right)$.
Note that $U\left(f;\rho,0\right)=U_{C}\left(f;\rho\right)$. The
prerequisite and \propref{icad} together give $U_{C}\left(f;\rho_{1}\right)-U_{C}\left(g;\rho_{1}\right)<0$
and $U_{C}\left(f;\rho_{2}\right)-U_{C}\left(g;\rho_{2}\right)>0$.
Take any $\epsilon\in\left(U_{C}\left(f;\rho_{1}\right)-U_{C}\left(g;\rho_{1}\right),0\right)$
such that $\left|\epsilon\right|<\frac{1}{2}\left[\mathbb{I}\left(f,\mu\right)-\mathbb{I}\left(g,\mu\right)\right]$.
Since $U_{C}\left(f;\rho\right)-U_{C}\left(g;\rho\right)$ is continuous
in $\rho$, there exists $\bar{\rho}$ such that $U_{C}\left(f;\bar{\rho}\right)-U_{C}\left(g;\bar{\rho}\right)=\epsilon$
(Intermediate Value Theorem). Fix $\bar{\rho}$ Let $\iota_{1}=0$,
then $U\left(f;\bar{\rho},\iota_{1}\right)-U\left(g;\bar{\rho},\iota_{1}\right)=\epsilon<0$,
so $g\succ_{\left(\mu,\bar{\rho},\iota_{1}\right)}f$ as desired.
Since $u_{\bar{\rho}}\left(\iota\right)$ is continuous and strictly
increasing in $\iota$, $u_{\bar{\rho}}\left(0\right)=0$, and $\lim_{\iota\rightarrow\infty}u_{\bar{\rho}}\left(\iota\right)=1$,
there exists $\iota_{2}>0$ such that $u_{\bar{\rho}}\left(\iota_{2}\right)=\frac{1}{2}$
(Intermediate Value Theorem). Fix $\iota_{2}$. Observe that $U\left(f;\bar{\rho},\iota_{2}\right)-U\left(g;\bar{\rho},\iota_{2}\right)=\left[U_{C}\left(f;\bar{\rho}\right)-U_{C}\left(g;\bar{\rho}\right)\right]+\frac{1}{2}\left[\mathbb{I}\left(f,\mu\right)-\mathbb{I}\left(g,\mu\right)\right]=\epsilon+\frac{1}{2}\left[\mathbb{I}\left(f,\mu\right)-\mathbb{I}\left(g,\mu\right)\right]>0$,
so $f\succ_{\left(\mu,\bar{\rho},\iota_{2}\right)}g$ as desired.

\paragraph*{Only if:}

Take $\mu$ as given. Inherit $U_{C}\left(f;\rho\right)$ and $U\left(f;\rho,\iota\right)$
from the ``if'' proof above. Due to \propref{icad}, $\left[U\left(f;\bar{\rho},\iota_{2}\right)-U\left(g;\bar{\rho},\iota_{2}\right)\right]-\left[U\left(f;\bar{\rho},\iota_{1}\right)-U\left(g;\bar{\rho},\iota_{1}\right)\right]=\left[u_{\bar{\rho}}\left(\iota_{2}\right)-u_{\bar{\rho}}\left(\iota_{1}\right)\right]\left[\mathbb{I}\left(f,\mu\right)-\mathbb{I}\left(g,\mu\right)\right]>0$.
Since $u_{\bar{\rho}}$ is strictly increasing, $\iota_{2}>\iota_{1}$.
\end{document}